\documentclass[conference]{IEEEtran}
\IEEEoverridecommandlockouts
\usepackage{amsmath,amssymb,amsfonts}
\usepackage{graphicx}
\usepackage{balance}  
\usepackage{subfig}
\usepackage{float}
\usepackage{fontspec}
\usepackage{bbding}
\usepackage{bm}
\usepackage{color,xcolor}
\usepackage{algorithm}
\usepackage[noend]{algpseudocode}

\newcommand{\revision}[1]{\textcolor{black}{#1}}

\newcommand{\best}[1]{\underline{\textbf{#1}}}
\newcommand{\secondbest}[1]{\textbf{#1}}

\usepackage{xspace}
\newcommand{\kw}[1]{{\ensuremath {\mathsf{#1}}}\xspace}

\newcommand{\sysname}{\kw{NeutronRT\xspace}}

\newcommand{\RTEC}{\kw{RTEC}}
\newcommand{\ODEC}{\kw{ODEC}}
\newcommand{\RTECNS}{\kw{RTEC}-NS\xspace}
\newcommand{\RTECNSTen}{\kw{RTEC}-NS10\xspace}
\newcommand{\RTECNSFive}{\kw{RTEC}-NS5\xspace}
\newcommand{\RTECUER}{\kw{RTEC}-UER\xspace}
\newcommand{\RTECFull}{\kw{RTEC}-Full\xspace}
\newcommand{\NrtInc}{\kw{NrtInc}}
\newcommand{\NrtIncC}{\kw{NrtInc(c)}}

\newcommand{\func}[1]{\textbf{\texttt{#1}}\xspace}
\newcommand{\algcomment}[1]{\textcolor{gray}{\emph{#1}}\xspace}
\usepackage{graphics}
\usepackage{epsfig}
\usepackage{enumitem}
\usepackage{amsfonts} 
\usepackage{libertine}
\usepackage{pifont}
\newcommand{\Paragraph} [1] {\smallskip\noindent{\bf #1. }}
\usepackage{booktabs}
\usepackage{multirow}
\usepackage[misc]{ifsym} 
\usepackage{url}
\usepackage[font=small,labelfont=bf]{caption}

\newtheorem{theorem}{Theorem}

\usepackage{tcolorbox}
\usepackage{listings}
\def\BibTeX{{\rm B\kern-.05em{\sc i\kern-.025em b}\kern-.08em
    T\kern-.1667em\lower.7ex\hbox{E}\kern-.125emX}}

\begin{document}



\title{Incremental GNN Embedding Computation on Streaming Graphs}


 \author{\IEEEauthorblockN{Qiange Wang$^{1}$, Haoran Lv$^{2}$, Yanfeng Zhang$^{2}$, Weng-Fai Wong$^{1}$, Bingsheng He$^{1}$}
\IEEEauthorblockA{$^{1}$\textit{National University of Singapore}, Singapore\\$^{2}$\textit{School of Computer Science and Engineering, Northeastern University}, Shenyang, China \\
wangqiange94@gmail.com;\{2401878@stu,zhangyf@mail\}.neu.edu.cn;
 \{dscwwf,dcsheb\}@nus.edu.sg;}
}




\maketitle
\begin{abstract}

Graph Neural Network (GNN) on streaming graphs has gained increasing popularity. However, its practical deployment remains challenging, as the inference process relies on Runtime Embedding Computation (\RTEC) to capture recent graph changes. This process incurs heavyweight multi-hop graph traversal overhead, which significantly undermines computation efficiency. We observe that the intermediate results for large portions of the graph remain unchanged during graph evolution, and thus redundant computations can be effectively eliminated through carefully designed incremental methods. In this work, we propose an efficient framework for incrementalizing \RTEC on streaming graphs. The key idea is to decouple GNN computation into a set of generalized, fine-grained operators and safely reorder them, transforming the expensive full-neighbor GNN computation into a more efficient form over the affected subgraph. With this design, our framework preserves the semantics and accuracy of the original full-neighbor computation while supporting a wide range of GNN models with complex message-passing patterns. To further scale to graphs with massive historical results, we develop a GPU–CPU co-processing system that offloads embeddings to CPU memory with communication-optimized scheduling. Experiments across diverse graph sizes and GNN models show that our method reduces computation by 64\%–99\% and achieves 1.7x–145.8x speedups over existing solutions.

\end{abstract}
\section{Introduction}
\label{sec:intro}

Graph neural networks (GNNs) have gained significant popularity for their effectiveness in modeling graph-structured data~\cite{taobao_cikm_2018, pinterest_kdd_2018, NEUTRONSTAR_SIGMOD_2022, gnn_survey_arxiv2019}. However, many real-world applications involve evolving graphs, which require GNN systems to promptly update embeddings and prediction results according to graph changes~\cite{REALTIMEGNN_VLDB_2021}. For example, short-video platforms~\cite{BYTEGNN_VLDB_2022, grale_kdd_2020} aim to incorporate real-time user–video interactions into user embeddings for recommendation and content moderation, while online financial services~\cite{inferturbo_ICDE_2023, taobao_cikm_2018, meituan} analyze the latest transactions to detect money laundering and malicious accounts. 

Recently, \textbf{Runtime Embedding Computation} (\RTEC)~\cite{Grapher_WWW_2024, REALTIMEGNN_VLDB_2021, Helios_PPOPP_2025, INKSTREAM_ARXIV_2023, ripple_arxiv_2025} has emerged as a promising solution for efficiently serving GNNs on streaming graphs. By updating embeddings with pre-trained models only for affected vertices, \RTEC can rapidly incorporate structural and feature changes from the most recent graph snapshots into prediction results. Compared with widely adopted periodical retraining-and-recomputation approaches~\cite{pinterest_kdd_2018, meituan, taobao_cikm_2018, inferturbo_ICDE_2023}, it delivers high-quality inference results.

However, high-performance deployment of \RTEC remains challenging, as it relies on a heavyweight inference-time $L$-layer GNN computation. The core inefficiency arises from the neighbor explosion problem~\cite{ROC}, where updating a single vertex can trigger computation across its entire $L$-hop neighborhood, as illustrated in Figure~\ref{fig:sec2:comp_amplification}.a-c. Consequently, \RTEC still incurs substantial overhead on real-world graphs: even with as little as 0.1\% of edge updates, the cost can reach up to 80\% of full-graph recomputation. We observe that \RTEC performs substantial redundant computation on subgraphs whose results remain valid. In an $L$-layer GNN, graph updates typically propagate up to $L$ hops, forming an \textit{affected subgraph} that consists of both the update propagation paths and the final-layer affected vertices (Figure~\ref{fig:sec2:comp_amplification}.b). Naive \RTEC recomputes not only the update propagation paths (in red) but also the \textit{unaffected subgraph}, whose intermediate results remain valid within the $L$-hop neighborhood (in blue). As shown in Figure~\ref{fig:sec2:comp_amplification}.c, updating just two edges ($\langle4,2\rangle$ and $\langle4,4\rangle$) triggers recomputation of $v_2$ using all nine edges, even though the contributions of seven edges have already been incorporated into the embeddings of $v_2$ (hop 2) and $v_4$ (hop 1). This inefficiency is further amplified on real-world graphs, where the redundant computation can account for up to 95\% of the total overhead (Section \ref{sec3:analysis}-A).

\begin{figure}[!t]
	\centering
	\includegraphics[width=1.02\linewidth]{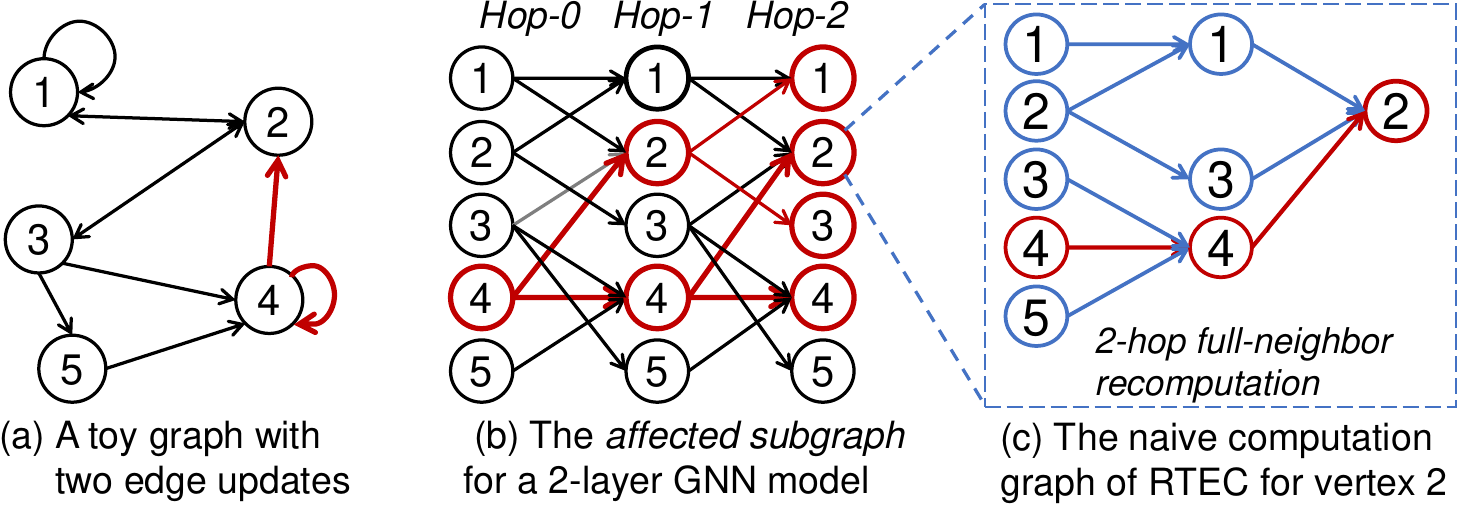}
    \vspace{-0.1in}
	\caption{\RTEC on a toy graph with edge updates $\langle4,2\rangle$ and $\langle4,4\rangle$ using a 2-layer model. Blue items in (c) indicate redundant computations on the \textit{unaffected subgraph}.}
\label{fig:sec2:comp_amplification}
    \vspace{-0.1in}
\end{figure}

To address this inefficiency, an intuitive approach is to leverage incremental processing, which eliminates redundant computation by transforming full recomputation into an equivalent yet more efficient procedure that reuses previously computed results. While such techniques have demonstrated substantial performance gains in graph processing and database analytics~\cite{gong2021ingress,vora2017kickstarter,LINVIEW_SIGMOD_2014,IVMTLF_SIGMOD_2018,pg_ivm}, extending them to GNN embedding computation remains challenging. This is because GNNs exhibit diverse and sophisticated computation patterns involving complex message-passing mechanisms and non-linear neural network operations. Traditional incremental methods are typically designed for analytical tasks with simple arithmetic operations, whose correctness assumptions do not naturally hold for GNN models, thereby rendering existing approaches unsuitable for GNN embedding computation (Section~\ref{sec2:challenge}).

In this work, we propose an efficient and generalized GNN \RTEC framework that enables incremental processing of \RTEC across diverse GNN models with theoretical accuracy guarantees. The framework finely decouples GNN computation into three components: neighbor-wise computation, aggregation of neighbor embeddings, and edge-level message computation, and executes them in a unified and safely reordered workflow. Under this workflow, each component’s output on \textit{affected subgraph} is updated incrementally by reusing results from its \textit{unaffected subgraph}. We further derive sufficient conditions under which the reordering is sound, ensuring that the final result is equivalent to that of full-neighborhood recomputation even in the presence of complex edge-level message dependencies. This design broadens applicability to GNN families previously considered incompatible, making real-time inference on streaming graphs practical and efficient, and delivering superior performance over state-of-the-art non-incremental methods \cite{Helios_PPOPP_2025,Grapher_WWW_2024}.

Incremental \RTEC requires maintaining historical results across layers~\cite{MEMFOOT_ATC_2021, GraphBOLT_EUROSYS_2019}. To efficiently support GNN \RTEC on large-scale graphs, we develop \sysname, a CPU–GPU co-processing system that offloads intermediate embeddings to CPU memory, performs computation on GPUs, and employs a communication-efficient scheduling mechanism to minimize transfer overhead. Experimental results show that \sysname can efficiently process billion-scale streaming graphs on a single NVIDIA A5000 GPU, reducing redundant computation by 64\%–99\% and achieving speedups ranging from 1.7x–145.8x over naive full-neighbor \RTEC and other non-incremental solutions, including neighbor sampling \cite{Helios_PPOPP_2025}, and directed embedding reuse \cite{Grapher_WWW_2024}.

In summary, we make the following contributions. 

\begin{itemize}[leftmargin=*]

\item We propose an general and efficient incremental \RTEC framework that transforms heavyweight full-neighbor computation into an equivalent yet more efficient form with fine-grained operator decoupling and reordering.

\item We formally establish the equivalence between incremental \RTEC and full-neighborhood computation under sufficient conditions, and demonstrate the theoretical correctness of our framework through illustrative examples.

\item We develop a GPU-CPU co-processing system that enables GPU-accelerated \RTEC on large-scale streaming graphs, achieving 1.7$\times$–145.8$\times$ speedup over state-of-the-art baselines across diverse GNN models and datasets.

\end{itemize}

\section{Background}
\label{sec:2}

\subsection{GNN Basics}
\label{sec:2:basis}
GNNs take a graph with features associated to each vertex as input, producing a representation vector for each vertex by stacking multiple \textbf{message-aggregate-update} layers:

\footnotesize
\begin{align}
\label{eq:aggr}
\textbf{h}_v^l=\textbf{UPD}\big(\textbf{AGG}(\{\textbf{h}^{l-1}_u* \textbf{MSG}(\textbf{h}^{l-1}_u, \textbf{h}^{l-1}_v)| u\in N(v)\}), \textbf{h}^{l-1}_v\big)
\end{align}
\normalsize
$\textbf{h}_{v}^l$ represents the embedding/feature of $v$ in the $l$-th layer. The \textbf{MSG} function computes the message for each edge. The \textbf{AGG} function gathers message-applied layer embeddings to produce a neighborhood representation, which is fed into the \textbf{UPD} function to calculate the embedding in the $l$-th layer. 

\Paragraph{Examples} In GCN model \cite{GCN_ICLR_2017}, the \textbf{MSG} function is the normalized degree of $<$$u,v$$>$. The \textbf{AGG} function is a \texttt{sum()}. The \textbf{UPD} function involves a {linear} transformation and a non-linear activation.

\footnotesize
\begin{align}
\label{eq:gcn}
\textbf{h}_v^l=\sigma(W^l(\sum_{_{u\in N(v)}} \frac{1}{\sqrt{d_ud_v}}\textbf{h}^{l-1}_u))
\end{align} \normalsize
Other GNN models with similar computation patterns include GraphSAGE~\cite{GRAPHSAGE_NIPS_2017}, GIN~\cite{GIN_ICLR_2019}, and CommNet~\cite{COMMNET_NIPS_2016}.

In contrast, some other models \cite{GAT_ICLR_2018,GNNGOODSURVEY_TPAMI_2024,GGCN_ICLR_2016,RGAT_ARXIV_2019} use sophisticated message functions that involves neural network and neighborhood data synchronizations. For example, the GAT model \cite{GAT_ICLR_2018} introduces self-attention mechanisms (\texttt{softmax()}) to distinguish important neighborhood. The message function is as follows:
\begin{align}\footnotesize\textstyle
\label{eq:gat:message}
\textbf{att}_{u,v}^l=\frac{\texttt{exp}(\texttt{LeakyReLU}(a^l([W^l \textbf{h}^{l-1}_v ||W^l \textbf{h}^{l-1}_u))}{\sum_{_{u\in N(v)}} \texttt{exp}(\texttt{LeakyReLU}(a^l[W^l \textbf{h}^{l-1}_v ||W^l \textbf{h}^{l-1}_u]))},
\end{align}
The \textbf{aggregate} and \textbf{update} functions in GAT involve simple \texttt{sum} aggregation and non-linear activations. %
{
\small
\begin{align}
\label{eq:gat:agg}
\textbf{h}_v^l=\sigma(\textstyle\sum_{_{u\in N(v)}} \textbf{att}_{u,v}^l W^l\textbf{h}_{u}^{l-1})
\end{align}
}
Similar models include A-GNN \cite{GNNGOODSURVEY_TPAMI_2024}, GGCN \cite{GGCN_ICLR_2016}, and RGAT\cite{RGAT_ARXIV_2019}. 

\subsection{GNN Training and Inference for Streaming Graphs}
On static graphs, GNNs are trained over multiple epochs, followed by a forward \textbf{embedding computation} pass that computes final-layer embeddings for downstream inference.



{
Streaming graphs evolve continuously through edge and vertex updates and support high-value applications such as real-time recommendation and fraud detection. These applications require timely and accurate node embeddings that reflect recent structural and feature changes~\cite{Helios_PPOPP_2025}. However, retraining models and recomputing embeddings over the entire graph for each update batch is computationally prohibitive. As a result, industrial systems often rely on periodic recomputation on snapshot graphs~\cite{meituan,grale_kdd_2020,taobao_cikm_2018}, which reduces overhead but fails to capture time-sensitive interactions, potentially leading to incorrect recommendations or classifications  for hundreds of thousands of users.



}

\section{\RTEC for Streaming Graphs}
\label{sec3}

\label{sec3:analysis}

Recent work has proposed \textbf{\texttt R}un\textbf{\texttt T}ime \textbf{\texttt E}mbedding \textbf{\texttt C}omputation (\RTEC)~\cite{Grapher_WWW_2024, REALTIMEGNN_VLDB_2021, INKSTREAM_ARXIV_2023, Helios_PPOPP_2025, ripple_arxiv_2025}, which continuously identifies and recomputes affected vertex embeddings using pre-trained models and evolving graph structures, achieving higher accuracy than periodic recomputation without model retraining (Section~\ref{sec:expr:accuracy}).

\begin{figure}[!t]
\vspace{-0.05in}
	\centering
	\includegraphics[width=1.0\linewidth]{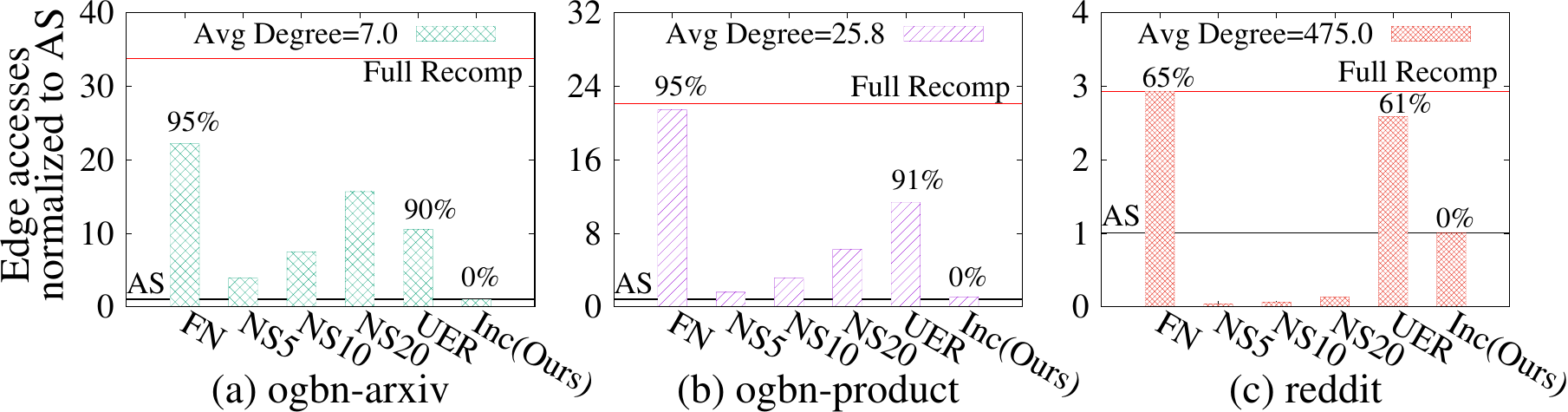}
   \vspace{-0.2in}
	\caption{
    Processed edge volume is normalized to the \textit{Affected Subgraph} (AS). FN, NS\#, and UER denote \RTEC with full-neighbor computation, neighbor sampling, and unaffected embedding reuse, respectively. The percentage of redundant computation on \textit{unaffected subgraphs} is labeled above each bar, except for NS approaches, which compute on both sampled \textit{affected} and \textit{unaffected subgraphs}.}
 	\label{fig:sec2:comp_volume}
   \vspace{-0.15in}
\end{figure}

\subsection{The Redundant Computation of \RTEC}
\label{sec2:redundant_comp}

Although \RTEC achieves high accuracy, its deployment on streaming graphs is hindered by redundant computation, i.e., repeatedly processing valid results from \textit{unaffected subgraphs} (blue region in Figure~\ref{fig:sec2:comp_amplification}.c). To quantify this overhead, we compare the number of edges processed by the \RTEC computation graph with those in the \textit{affected subgraph} (AS, red region in Figure~\ref{fig:sec2:comp_amplification}.b–c). We split the most recent 10\% of edges into 100 batches and process each batch independently. Figure~\ref{fig:sec2:comp_volume} reports the average processed-edge volume per batch, normalized to AS. Even when only 0.1\% of edges are updated, AS typically contains just 2.5\%–15\% of the original graph, while naive full-neighbor \RTEC (FN) processes 2.9$\times$–22.2$\times$ more edges than AS, approaching full-graph recomputation on large graphs. Notably, 65\%–95\% of this computation is spent on unaffected subgraphs.

\begin{figure}[!t]
	\centering
	\includegraphics[width=\linewidth]{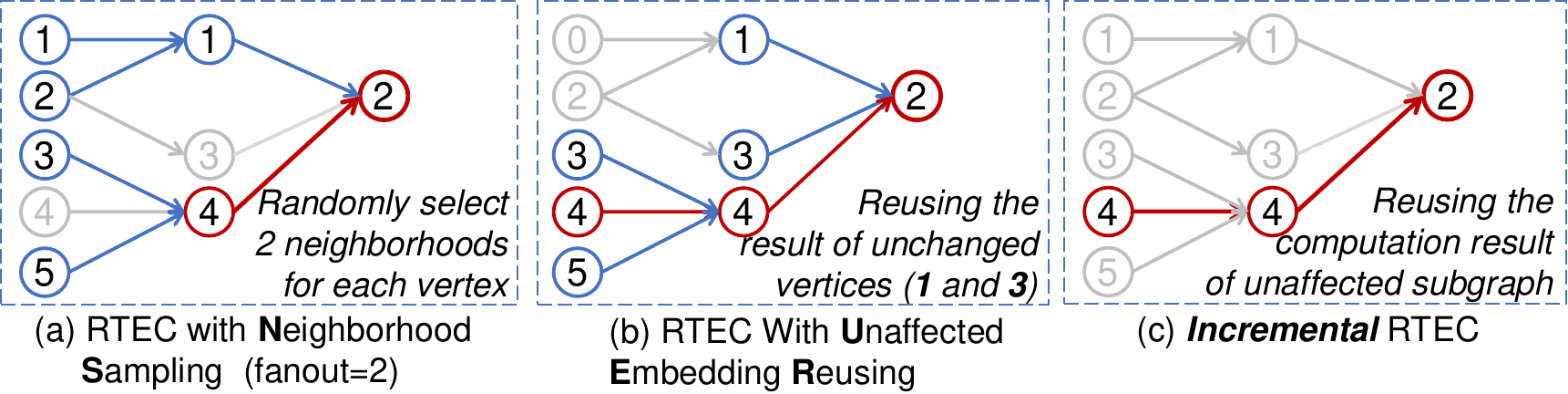}
    \vspace{-0.2in}
	\caption{An illustration of \RTECNS and \RTECUER and \RTEC-Inc for the affected vertex v2. \RTECNS and \RTECUER cannot eliminate redundant computation on \textit{unaffected subgraphs}.}
 	\label{fig:sec2:comp_amplification2}
    \vspace{-0.1in}
\end{figure}

\subsection{Limitation of Non-incremental Solutions}

\label{sec2:relatedwork}
\Paragraph{\RTEC with Neighborhood Sampling (NS)}
Neighborhood sampling-based \RTEC~\cite{Helios_PPOPP_2025, DGL_ARXIV_2019, GRAPHSAGE_NIPS_2017} reduces computation by randomly dropping a subset of the $L$-hop neighborhoods of affected vertices (Fig.~\ref{fig:sec2:comp_amplification2}.a). However, its effectiveness is highly data-dependent. As shown in Fig.~\ref{fig:sec2:comp_volume}, varying the sampling fanout from 5 to 20 leads to diverse outcomes. On high-degree graphs (e.g., Reddit), sampling can reduce computation below the size of the \textit{affected subgraph} (\textbf{AS}), while on low- and medium-degree graphs, the processed edge volume remains high even with small fanouts. Moreover, random neighbor dropping may cause significant accuracy degradation, since only a subset of the affected subgraph is evaluated. As shown in Section~\ref{sec:expr:accuracy}, small fanouts can yield lower accuracy than periodic recomputation, as the information loss from sampling may outweigh the benefit of fresher graph snapshots.

\Paragraph{\RTEC with unaffected embedding reuse (UER)}
Another approach is to cache and reuse embeddings of unaffected vertices~\cite{Grapher_WWW_2024,HONGTU_SIGMOD_2024}.  However, this vertex-centric approach still performs full-neighbor aggregation for affected vertices, even when only a single incoming edge changes. As illustrated in Figure~\ref{fig:sec2:comp_amplification2}.b, updating the embedding of vertex~2 requires processing seven edges, four of which are unaffected. As shown in Figure~\ref{fig:sec2:comp_volume}, \RTECUER processes 2.5$\times$–12.3$\times$ more edges than the \textit{affected subgraph}, with redundant computation on unaffected edges still accounting for 61\%–91\% of the total. As summarized in Table~\ref{tab:sec2:related}, \RTECNS and \RTECUER either sacrifice inference accuracy or incur significant redundant computation, making them unsuitable for RTEC applications.



\begin{table}
\footnotesize
\vspace{-0.05in}
\renewcommand{\arraystretch}{1}
\setlength{\tabcolsep}{3pt}
\centering
	\caption{Summary of existing \RTEC solutions.}
\vspace{-0.05in}
	\label{tab:sec2:related}
\begin{tabular}{l |l| r | r| r  l  r  r r } 
\hline

\multirow{2}*{Approach}&\multirow{2}*{Baseline} &Redundant & Accuracy& Complex \\ 
&    &Computation & Guarantee & Model Support \\ \hline
Non-&SMP \cite{Helios_PPOPP_2025}&Unstable&$\times $&$\checkmark$\\	
Incremental& UER \cite{Grapher_WWW_2024}&High&$\checkmark $&$\checkmark$\\	
\hline
\multirow{2}*{{Incremental}}&Naive \cite{INKSTREAM_ARXIV_2023,ripple_arxiv_2025}&Low&\textcolor{blue}{$\checkmark$(simple GNN)}&$\times$\\
&\textbf{Our work}&\textbf{Low}&\textbf{$\checkmark $}&\textbf{$\checkmark$}\\	
\hline
\end{tabular}
\vspace{-0.05in} 
\end{table}

\begin{figure}
    \vspace{-0.05in}
	\centering
	\includegraphics[width=1.0\linewidth]{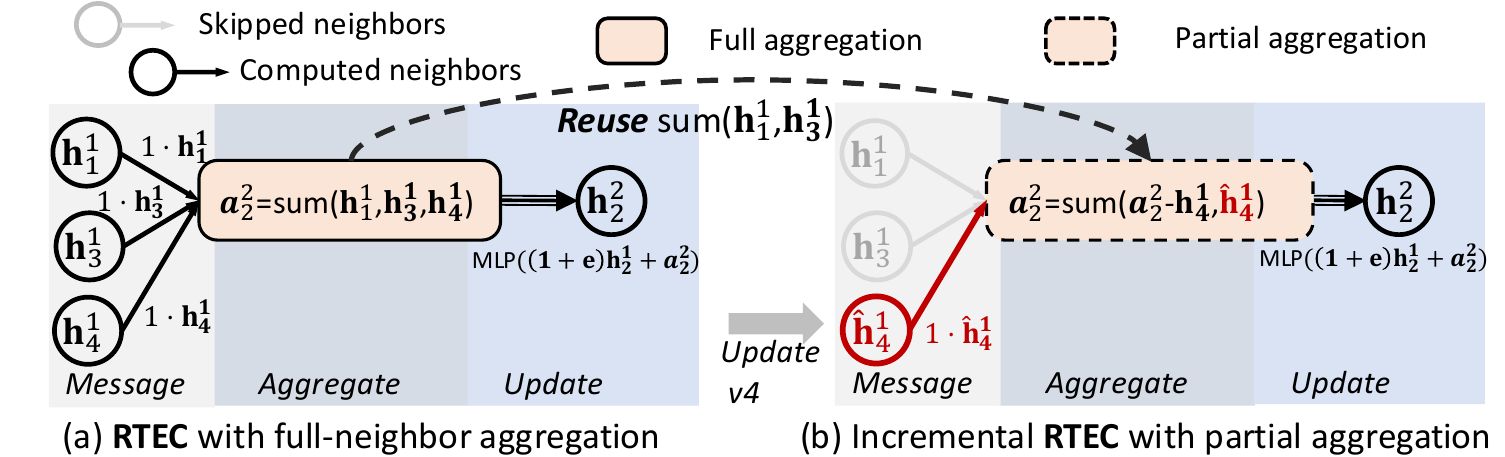}
        \vspace{-0.15in}
	\caption{Full-neighbor \RTEC and incremental \RTEC for the GIN model. \textbf{Message}: $f(x)=x$; \textbf{Aggregate}:\texttt{sum()}; \textbf{update}: \texttt{MLP()}.}
 	\label{fig:sec3:comp_amplification}
    \vspace{-0.23in}
\end{figure}
{
\subsection{Incremental \RTEC: Opportunity and Challenges}

\label{sec2:challenge}
\Paragraph{Opportunity}
Incremental processing reduces redundant computation by reusing historical results from unaffected regions and updating only newly affected parts. Figure~\ref{fig:sec3:comp_amplification} illustrates this intuition using the GIN model~\cite{GCN_ICLR_2017} with \texttt{sum} aggregation and constant edge messages. When vertex $v_4$ is updated, only the message on edge $\langle v_4, v_2 \rangle$ needs to be recomputed as $1\cdot\hat{\textbf{h}}^1_4$. The valid contributions from unaffected edges $\langle v_1, v_2 \rangle$ and $\langle v_3, v_2 \rangle$ are reused by subtracting the outdated message $\textbf{h}^1_4$ from the previous aggregation $\textbf{a}^2_2$ and adding the updated one. The new embedding $\textbf{h}^2_2$ is then obtained by applying the MLP to the updated aggregation. This process recomputes only the affected edge while producing results equivalent to full recomputation, enabled by 1) the associativity of the \texttt{sum} operator ensures \(\texttt{sum}(\textbf{h}^1_1,\textbf{h}^1_3,\hat{\textbf{h}}^1_4) = \texttt{sum}(\texttt{sum}(\textbf{h}^1_1, \textbf{h}^1_3,\textbf{h}^1_4)-\textbf{h}^1_4,\, \hat{\textbf{h}}^1_4),
\), and 2) the validity of unaffected edge contributions (e.g., $1\cdot\textbf{h}^1_1$, $1\cdot\textbf{h}^1_3$).

{

\Paragraph{Overhead Comparison}
Incremental RTEC (RTEC-Inc) restricts propagation to the affected subgraph and can be viewed as performing an $L$-hop broadcast from the updated vertices, with an edge computation volume of $O(d\,|V_{\text{upd}}|\cdot \alpha^{L+1})$, where $d$ is the average degree and $\alpha$ denotes the average number of affected neighborhood per layer; $\alpha$ varies with graph size, degree distribution, skewness, and the distribution of updates. In contrast, \RTECFull reprocesses the $L$-hop neighborhoods of the final-layer affected vertices, performing an additional $L$-hop broadcast over RTEC-Inc and incurring a computation volume of $O(d\,|V_{\text{upd}}|\cdot \alpha^{2L+1})$. \RTECUER avoids recomputing unaffected vertices but still incurs an $(L{+}1)$-hop broadcast with cost $O(d\,|V_{\text{upd}}|\cdot \alpha^{L+2})$. In general, when update batches are small ($|V_{\text{upd}}|$) and $\alpha$ is large, RTEC-Inc achieves orders-of-magnitude lower computation than \RTECFull and \RTECUER. When updates span the entire graph, all three methods degenerate to full-graph computation. Sampling-based methods (\RTECNS) follow a similar $2L$-hop propagation pattern but operate on a sampled subgraph of size $O(d,|V_{\text{upd}}|\cdot \hat{\alpha}^{2L+1})$, where $\hat{\alpha}$ is bounded by the sampling fanout; consequently, on high-degree graphs (e.g., Reddit), \RTECNS can even incur lower computation overhead than RTEC-Inc, as shown in Figure~{8}.

}

{
\Paragraph{Related Work} Despite the advantages, general incremental GNN embedding computation remains highly challenging.
Recent efforts~\cite{INKSTREAM_ARXIV_2023,ripple_arxiv_2025} study incremental GNN embedding computation for models with constant edge message functions and simple aggregation operators (\texttt{min}, \texttt{max}, \texttt{sum}), such as GraphConv, GraphSAGE-sum, and GIN, where recomputation over unaffected regions (e.g., the blue area in Fig.~1) can be avoided. However, their applicability becomes limited when edge message computation depends on neighborhood information. A representative example is GCN: although degree normalization appears constant, it changes dynamically as the graph evolves, causing updates to propagate to all incident edges of a neighbor. Correctly identifying and handling such dependencies across different GNN architectures is non-trivial and cannot be directly supported by existing execution models. As a result, prior systems fall back to full-neighbor recomputation for more complex models (e.g., GCN and GAT) to ensure correctness~\cite{INKSTREAM_ARXIV_2023,ripple_arxiv_2025}. 

\Paragraph{Challenges}
The core challenge arises from the mismatch between the diverse GNN computation patterns and the strict applicability conditions.
Specifically, two requirements are hard to satisfy simultaneously.
First, the neighborhood \textbf{aggregation} operation must be associative, i.e., $f(x,y,z)=f(f(x,y),z)$, to allow new messages to allow new messages to be combined incrementally with historical results. Second, the intermediate embeddings to be reused must remain valid, meaning that messages from unaffected neighborhoods do not need to be recomputed.
These conditions are violated in many common GNN models.
For example, GraphSAGE-mean employs non-associative \kw{mean} aggregation, where
\(\kw{mean}(\textbf{h}_x,\, \textbf{h}_y,\, \textbf{h}_z) \ne \kw{mean}(\kw{mean}(\textbf{h}_x,\, \textbf{h}_y),\, \textbf{h}_z)\). In GCN and GAT, neighbor-dependent contexts such as degree normalization and attention weights are coupled with message computation; updates to these contexts can invalidate all neighbor messages and prevent reuse.

}
}

{
\section{Incremental \RTEC Framework}
In this work, we propose a unified and fine-grained GNN abstraction that decouples neighborhood-wise context computation from message and aggregation operations, enabling unified and correct incremental processing across both simple and complex GNN architectures with theoretical guarantee.
}

\label{sec:frame}
\subsection{Fine-grained Operator Decoupling for Incremental \RTEC}

We formally describe the abstraction as follows.
\small
\begin{align}
\label{eq:decouple}
\Big\{mlc_{uv}&=\textbf{ms\_local}\big(\textbf{h}^{l-1}_u, \textbf{h}^{l-1}_v\big) \mid u\in N(v)\Big\}\\
nct_{v}&=\textbf{nbr\_ctx}\big(\{mlc_{uv}| u\in N(v)\}\big)\\
\Big\{msg_{uv}&=\textbf{ms\_cbn}\big(nct_{v},mlc_{uv} \big) \mid u\in N(v)\Big\}\\
\textbf{a}_v^l&=\textbf{aggregate}(\big\{msg_{uv}\cdot\textbf{f\_nn}(\textbf{h}^{l-1}_u)\mid u\in N(v)\big\}) \\
\textbf{h}_v^l&=\textbf{update}(\textbf{a}_v^l)
\end{align}\normalsize
This abstraction extends the standard \textbf{message-aggregate-update} model by explicitly decoupling the neighbor-wise context computation. Specifically, $\textbf{ms\_local}()$ denotes an edge-wise message function that can be computed independently. $\textbf{nbr\_ctx}()$ represents the neighbor-wise context computation. It takes as input either edge-wise messages or constant values from all neighborhoods and produces a single value representing a neighborhood property. For example, in GCN, the destination vertex degree can be computed by summing a constant value of 1 from each neighboring vertex using \textbf{nbr\_ctx()}. The $\textbf{msg\_cbn}()$ function operates on each edge, combining the edge message with the neighbor-wise context to reproduce the original \textbf{message} semantics. The \textbf{aggregate} and \textbf{update} operation retain the same semantic as in the original form, where the $\textbf{f\_nn}(\textbf{h}^{l-1}_u)$ is a linear transformation function applying the message to the vertex feature, which can be constants or matrix computations. We provide a graphical example in Figure \ref{fig:sec3:condition}.a, where \textbf{f\_nn()} is implicitly included in the aggregation operation.

The key insight of this abstraction is that, by decoupling the neighbor-wise context, the \textbf{ms\_local}() retains only components that can be evaluated independently on each edge, and the aggregation operation can be transformed into an associative form, e.g., replacing \texttt{mean} with \texttt{sum} by separating the in-degree computation. The neighbor-wise context can also be incrementally updated if the \textbf{nbr\_ctx()} operation is associative. Furthermore, if its effect is independent of the local messages and separable from the aggregation operation, changes in the neighbor-wise context can be directly applied to the aggregated result, rather than recomputing all edge messages, via efficient vertex-wise \textbf{msg\_cbn()}. This decoupling enables these components to be reordered to support incremental processing on only affected neighborhoods. 

\begin{figure*}
    \vspace{-0.1in}
	\centering
	\includegraphics[width=1.0\linewidth]{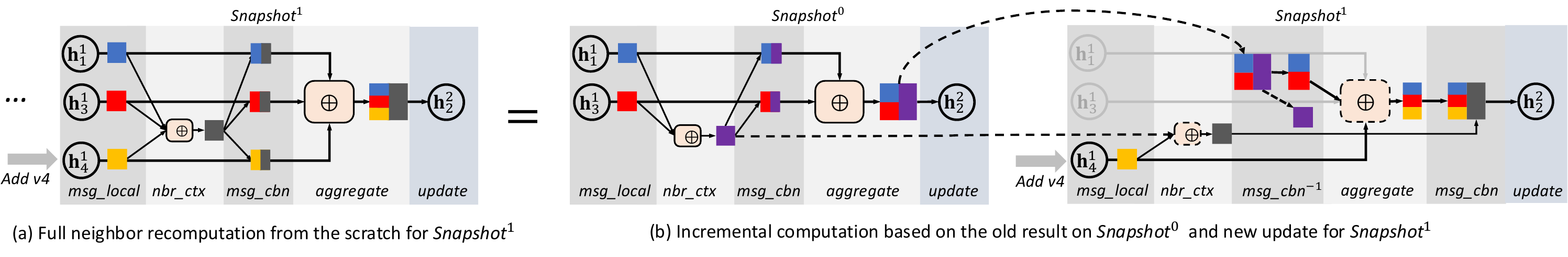}
        \vspace{-0.2in}
	\caption{A graphical illustration of full-neighbor \RTEC and the reordered incremental \RTEC. Colored boxes represent result tensors and their compositions from upstream operations; solid arrows indicate data flow, while dashed arrows denote the reuse paths.}
 	\label{fig:sec3:condition}
    \vspace{-0.10in}
\end{figure*}

{
\begin{algorithm}[t]\small
	\caption{Reordered incremental \RTEC for a single layer.}  
	\label{alg:inc_comp}  
	\begin{algorithmic}[1]
		\Require  Destination vertex $v$; Affected neighborhood $\Delta N(v)$, previous neighbor embedding $\textbf{a}_v^l$; previous neighbor context ${nct}_{v}$.
		\Ensure  New result ${h}^{l}_v$
        \For{each $u$ in $\Delta N(v)$}
        \State $mlc_{uv}=\textbf{ms\_local}\big(\textbf{h}^{l-1}_{u}, \textbf{h}^{l-1}_v\big)$
        \EndFor
        \State $\hat{nct}_{v}=\textbf{nbr\_ctx}({nct}_{v}, \big\{msg\_lc_{uv} \mid u \in \Delta N(v)\big\})$        
        \State $\hat{\textbf{a}}_v^l=\textbf{ms\_cbn}^{-1}\big({nct}_{v},\textbf{a}_v^l\big)$

\State $\hat{\textbf{a}}_v^l=\textbf{aggregate}(\hat{\textbf{a}}_v^l, \big\{mlc_{uv}\cdot \textbf{f\_nn}(\textbf{h}^{l-1}_{u}) \mid u\in \Delta N(v)\big\})$
\State $\textbf{a}_v^l=\textbf{ms\_cbn}\big(\hat{nct}_{v},\hat{\textbf{a}}_v^l,\big)$
\State  $\textbf{h}_v^l=\textbf{update}(\textbf{a}_v^l)$       
	\end{algorithmic} 
\end{algorithm}
}

\subsection{Reordered Incremental \RTEC Workflow}
Algorithm \ref{alg:inc_comp} shows the general form of incremental \RTEC for a single layer. The input includes a target vertex \(v\), the affected neighbors \(\Delta N(v)\), previous aggregated neighbor embedding \(\textbf{a}_v^l\), and the previous neighborhood context \(nct_v\). 
In the first step, \textbf{msg\_local} recomputes local messages $mlc_{uv}$ for affected edges (Line 1-2) and then \textbf{nbr\_ctx} partially computes the neighbor-wise context using the old \(nct_v\) and new local edge messages (Line 3). In the second step, local messages and new neighbor-wise context are separately applied to the neighbor aggregation embedding $\textbf{a}_v^l$. It first uses \textbf{ms\_cbn}()'s inverse operation \textbf{ms\_cbn}()$^{-1}$ to remove the effect of old \(nct_v\) from \(\textbf{a}^l\) (Line 4), then partially aggregates the update of local messages from affected neighborhoods to $\hat{\textbf{a}}^l$ (Line 5), reusing the valid aggregation result of unaffected subgraphs, and finally recalls \textbf{ms\_cbn}() to restore \({\textbf{a}}^l\) using the new neighbor-wise context \(\hat{nct}\) (Line 6). In the third step, the update operation recomputes the new layer-output embedding using the new \({\textbf{a}}^l\) (Line 7). Note that Algorithm~\ref{alg:inc_comp} presents a generalized formulation for different type of affected vertices. For each inserted element in $\Delta N(v)$, \textbf{ms\_local()} adds a positive message; for each deleted element, it produces a negative message to cancel the outdated contribution. Element updates are accomplished through a combination of both operations.

\subsection{Equivalence Analysis}
\label{sec3:condition}

Algorithm~\ref{alg:inc_comp} is equivalent to the original full-neighbor formulation (Equation 5-9). Figure~\ref{fig:sec3:condition} illustrates the key insight behind the equivalence between full-neighbor \RTEC (a) and its incremental counterpart (b). In addition to the associative property that enables the \textbf{nbr\_ctx}() and \textbf{aggregate}() operations to be computed incrementally. The key property enabling the algorithm is that the effects of neighborhood context and local messages are independent and separable from the aggregation result, i.e., \textbf{ms\_cbn}() is distributive over the \textbf{aggregate} operation (Condition 3 in theorem \ref{th:inc}). This allows us to use vertex-centric \(\textbf{msg\_cbn}^{-1}()\) and \(\textbf{msg\_cbn}()\) operations to remove outdated and apply updated neighborhood context to the aggregation result, and to independently accumulate local messages from affected edges (e.g., \(\langle 4, 2 \rangle\)) into the intermediate aggregation state. We formalize the correctness condition in the following theorem.

\begin{theorem}
\label{th:inc}\textbf{Equivalence of full-neighbor \RTEC and incremental \RTEC}
Given a vertex \(v\), its original neighborhood \(N\), an update set \(\Delta N\), and the corresponding vertex embeddings \(\{ \textbf{h}^l_u \mid u \in N \cup \Delta N \}\), the output embedding computed by Algorithm~\ref{alg:inc_comp}, using the incremental update over \(N\) and \(\Delta N\), is equivalent to the embedding obtained by recomputing from scratch using Equations (9)–(13) on the combined neighborhood \(N \cup \Delta N\), provided the following conditions hold. $M$ and $X$ represent a set of edge messages and vertex embeddings from the domain of \textbf{nbr\_ctx} and \textbf{aggregate}, respectively.

\begin{itemize}[leftmargin=*]
    \item (1) $\textbf{nbr\_ctx}(M_l\cup M_r) = \textbf{nbr\_ctx}\big(\textbf{nbr\_ctx}(M_{l}), M_r\big)$
    \item (2) $\textbf{aggregate}(X_l\cup X_r) = \textbf{aggregate}\big(\textbf{aggregate}(X_{l}), X_r\big)$
    \item (3) $\textbf{aggregate}(\{\textbf{msg\_cbn}(z,m)\mid m \in M\})=$\\
    $\textcolor{white}{1}\quad\quad\quad\quad\quad\quad\textbf{msg\_cbn}(z, \textbf{aggregate}(\{m\mid m \in M\}))$
    \item (4) $\forall z_1, z_2\in M, z_1 \le z_2 \Rightarrow \textbf{msg\_cbn}(z_1,m) \preceq \textbf{msg\_cbn}(z_2,m)$
\end{itemize}

\end{theorem}
\begin{IEEEproof}\label{proof:correct2}
In the proof, we use the superscript notation \(^{[]}\) to distinguish the input data of neighborhood context and \textbf{aggregate} result. For brevity, we denote \(N \cup \Delta N\) as \(N^+\). Since \textbf{msg\_local()} is edge-wise, full-neighbor \RTEC and incremental \RTEC produce the same output messages. Then, the full neighbor-wise computation (Equation 10) can be converted into the partial computation form in Line 3 using condition (1):

\footnotesize
\begin{align}
\label{proof:nbr}
nct^{[N^+]}_{v}=&\textbf{nbr\_cxt}(\{mlc_{uv} \mid u \in N^+\})\notag\\
=&\textbf{nbr\_cxt}(\textbf{nbr\_cxt}(\{mlc_{uv} \mid u \in N) \cup \{mlc_{uv} \mid u \in \Delta N\}) \notag\\
=&\textbf{nbr\_cxt}(nct^{[N]}_{v}, \{mlc_{uv} \mid u \in \Delta N\})
\end{align}
\normalsize

Condition 2, similar to Condition 1, ensures that partial aggregation in Line 5 yields the same outcome as full aggregation over complete data. Condition 3 indicates that the \textbf{msg\_cbn}() function is distributive over the $\textbf{aggregate}()$ operation. Such that Equations 11-12 can be rewrite as follow:

\footnotesize
\begin{align}
\textbf{a}^{l,[N^+]}_v&=\textbf{aggregate}(\big\{\textbf{ms\_cbn}\big(nct_{v}^{[N^+]},mlc_{uv}\cdot\textbf{f\_nn}(\textbf{h}^{l-1}_u)\big)\mid u\in N^+\big\})\notag\\
&=\textbf{ms\_cbn}\big(nct_{v}^{[N^+]},\textbf{aggregate}(\big\{mlc_{uv}\cdot\textbf{f\_nn}(\textbf{h}^{l-1}_u)\mid u\in N^+\big\})\big)\notag\\
&=\textbf{ms\_cbn}(nct_{v}^{[N^+]},\textbf{aggregate}\big(\notag\\
&\quad\quad\quad\quad\quad\quad\textbf{aggregate}(\big\{mlc_{uv}\cdot\textbf{f\_nn}(\textbf{h}^{l-1}_u)\mid  u\in N\big\})\notag\\&\quad\quad\quad\quad\quad\quad\quad\quad\quad\cup\big\{mlc_{uv}\cdot\textbf{f\_nn}(\textbf{h}^{l-1}_u)\mid  u\in\Delta N\big\}\big)
\end{align}
\normalsize
Condition 4 is a sufficient condition for the existence of the inverse function $\textbf{msg\_cbn}^{-1}$(), with which we can decouple the effect of neighborhood-wise computation from the aggregation result. Such that, the last \textbf{aggregate(...)} in Equation 15 can be rewrote as follows:

\footnotesize
\begin{align}
&\textbf{aggregate}(\big\{mlc_{uv}\cdot\textbf{f\_nn}(\textbf{h}^{l-1}_u)\mid  u\in N\big\})\notag\\
=&\textbf{msg\_cbn}^{-1}\big(nct^{[N]}_{v}, \textbf{msg\_cbn}( nct^{[N]}_{v},\notag\\
&\quad\quad\quad\quad\quad\quad\quad\textbf{aggregate}(\big\{mlc_{uv}\cdot\textbf{f\_nn}(\textbf{h}^{l-1}_u)\mid  u\in U\big\}))\big)\notag\\ 
=&\textbf{msg\_cbn}^{-1}\big(nct^{[N]}_{v},\notag\\
&\quad\quad\textbf{aggregate}(\big\{\textbf{msg\_cbn}\big( nct^{[U]}_{v},mlc_{uv}\cdot\textbf{f\_nn}(\textbf{h}^{l-1}_u)\big)\mid  u\in N\big\}\big)\big)\notag\\
=&\textbf{ms\_cbn}^{-1}\big({nct}^{[N]}_{v},\textbf{a}_v^{l,[N]}\big)
\end{align}
\normalsize
Finally, Equation 15 can be rewrote with Equation 16 as follows:
\footnotesize
\begin{align}
&\textbf{a}^{l,[N^+]}_v=\textbf{ms\_cbn}(nct_{v}^{[N^+]},\textbf{aggregate}\big(\notag\\
&\textbf{ms\_cbn}^{-1}({nct}^{[N]}_{v},\textbf{a}_v^{l,[N]})\cup \big\{mlc_{uv}\cdot\textbf{f\_nn}(\textbf{h}^{l-1}_u)\mid  u\in\Delta N\big\}\big)
\end{align}
\normalsize
This formulation aligns to Lines 4-6 of Algorithm \ref{alg:inc_comp}, with which \textbf{update}() computation produces the correct result. Proved.
\end{IEEEproof}

{
\Paragraph{GNN models with constrained incremental processing}

Theorem~\ref{th:inc} establishes the theoretical foundation for the correctness of incremental processing based on the properties of computational operators. {It implicitly assumes that the destination embedding $\textbf{h}^{l-1}_v$ does not participate in the \textbf{ms\_local()} function in Equation~(5), which ensures that results from unaffected regions remain correct for safe reuse.}. However, in GNN models where the message computation also involves destination embedding (e.g., $\textbf{h}^{l-1}_v$ in GAT), any destination embedding updates can affect the local messages of all its neighbors, leading to incorrect result reuse even when the conditions are satisfied. To guarantee correctness, our method recomputes embeddings for such destination-affected vertices using their full neighborhoods. Importantly, this constraint does not compromise the overall efficiency of incremental processing, as the number of such vertices is significantly smaller than that of vertices requiring incremental updates (Section~\ref{sec:expr:im}).

}

\begin{table*}[!ht]
\vspace{-0.1in}
	\caption{Representative GNN models adaptable to incremental computation. The layer superscript is omitted for brevity. MoNet and CommNet are inherently incremental models, while the remaining ones can be incrementalized with our design.}
	\vspace{-0.05in}
	\label{tab:alg:taxomi}
	\centering
	\scriptsize
	{\renewcommand{\arraystretch}{1.2}
    \setlength{\tabcolsep}{2pt}
	\begin{tabular}{l |r |r |r |r| r |r |r r r r }
		\hline
		
		\hline
		\multirow{2}*{\textbf{Model}} &
		\textbf{msg\_local$(\textbf{h}_u,\textbf{h}_v)$} &
		\textbf{nbr\_ctx$\big(\{mlc_{uv}$$\mid$$u$$\in$$U\}\big)$}  &
		\textbf{msg\_cbn($mlc_{uv},nct_v)$)}&
        \textbf{msg\_cbn$^{-1}$($msg_{uv},nct_v$)}&
        \textbf{aggregate($\cdot$)}&
		\textbf{f\_nn($\textbf{h}_u$)}&
		\textbf{update($\textbf{h}_v, \textbf{a}_v$)} \\
         
        &
        $=mlc_{uv}$ &
		$=nct_{v}$  &
		$=msg_{uv}$ &
        $=mlc_{uv}$ &
        $=\textbf{a}_v$&
		$\Rightarrow\textbf{agg}$&
		$=\textbf{h}_v$ \\
		\hline
        
        MoNet \cite{GNNGOODSURVEY_TPAMI_2024}&
        $\textbf{exp}(\frac{1}{2}(\textbf{h}_u$$-$$w_u)^t W_j(\textbf{h}_u$$-$$w_u))$ &
        1&
        $1\cdot mlc_{uv}$&
        $1\cdot msg_{uv}$&
        $\textbf{sum}(\cdot)$&
        1&
        $\textbf{ReLU}(W\textbf{a}_v)$ \\

        {CommNet \cite{COMMNET_NIPS_2016}} & 
        1 &
        1 &
        $1\cdot mlc_{uv}$&
        $1\cdot msg_{uv}$& 
        $\textbf{sum}(\cdot)$ & 
        $\textbf{h}_{u}$ &
        $W_1\textbf{h}_v+W_2\textbf{a}_v$\\
        \hline

        {GCN\cite{GIN_ICLR_2019}} &
        $\frac{1}{\sqrt{d_u}}$ & 
        $\textbf{count}(\cdot)$ &
        $mlc_{uv}\cdot \frac{1}{\sqrt{nct_v}}$&
        $msg_{uv}\cdot \sqrt{nct_v}$&
        $\textbf{sum}(\cdot)$& 
        $\textbf{h}_u$& 
        $\textbf{ReLU}(W\textbf{a}_v)$ \\
        GraphSAGE\cite{pinterest_kdd_2018} & 
        1& 
        $\textbf{count}(\cdot)$& 
        ${mlc_{uv}}\big/{nct_v}$  &
        $msg_{uv}\cdot nct_{v}$&  
        $\textbf{sum}(\cdot)$ &
        $\textbf{h}_{u}$ &
        $\textbf{ReLU}(W\textbf{a}_v)$ \\

PinSAGE\cite{pinterest_kdd_2018} & 
$\alpha_{u,v}{\sigma}(Q\textbf{h}_{u}+q)$ &
$\textbf{count}(\cdot)$ &
${mlc_{lc}}\big/{nct_v}$ &
$msg_{uv}\cdot nct_{v}$& 
$\textbf{sum}(\cdot)$ &
1 &
$\sigma(W(\textbf{h}_{v} \|\textbf{a}_v))$\\

{RGCN} \cite{RGCN_ESWC_2018}&
         $W_r$  & 
         $\textbf{count}(\cdot)$&
         ${mlc_{lc}}\big/{nct_v}$ &
         $msg_{uv}\cdot nct_{v}$ &
         $\textbf{sum}(\cdot)$ &
         $\textbf{h}_{u}$ &
         $\sigma(W_o\textbf{h}_v+\textbf{a}_v)$\\

	GAT\cite{GAT_ICLR_2018}& 
    $\textbf{exp}\big(\sigma(a([W \textbf{h}_v ||W \textbf{h}_{u})\big)$&
    $\textbf{sum}(\cdot)$ &
    ${mlc_{uv}}\big/{nct_v}$&
    $msg_{uv}\cdot nct_v $ &
    $\textbf{sum}(\cdot)$&
    $W\textbf{h}_u$&
    $\textbf{elu}(\textbf{a}_v)$\\
        
    {G-GCN\cite{GGCN_ICLR_2016} }&
    $\sigma(W_1\textbf{h}_u+W_2\textbf{h}_v)$ &
    1&
    $1\cdot mlc\_{uv}$ &
    $1\cdot msg_{uv}$& 
    $\textbf{sum}(\cdot)$ &
    $\textbf{h}_u$ &
    $\sigma(W\textbf{a}_v)$\\
    
	{A-GNN\cite{OGB_DATASET}}  &
    $w\frac{\textbf{h}_{v}^{T}\cdot\textbf{h}_{u}}{\|\textbf{h}_{v}\| \|\textbf{h}_{u}\|}$&
    1 &
    $1\cdot mlc_{uv}$&
    $1\cdot msg_{uv}$& 
    $\textbf{sum}(\cdot)$&
    $\textbf{h}_{u}$&
    $\sigma(W\textbf{a}_v)$ \\        

    RGAT\cite{RGAT_ARXIV_2019}& 
    $\textbf{exp}\big(\sigma(a_r([W_r \textbf{h}_v ||W_r \textbf{h}_{u})\big)$&
    $\textbf{sum}_{r\in R}(\cdot)$ &
    ${mlc_{uv}}\big/{nct_{v[R(u,v)]}}$&
    $msg_{uv}\cdot nct_{v[R( u,v)]}$ &
    $\textbf{sum}(\cdot)$&
    $W_r\textbf{h}_u$&
    $\sigma(\textbf{a}_v)$\\
    \hline
		
	\end{tabular}
	}
    \vspace{-0.05in}
\end{table*}

\vspace{-0.05in}
\subsection{Application to Various GNN Models}
We summarize ten representative incrementalizable GNN models in Table~\ref{tab:alg:taxomi} and demonstrate its practical application using GAT~\cite{GAT_ICLR_2018} as a representative example.

\begin{algorithm}[t]\small
	\caption{Full-neighbor \RTEC for GAT}  
	\label{alg:gat1}  
	\begin{algorithmic}[1]
		\Require  
         Input vertex $v$; Affected neighborhood $\Delta N(v)$; The new embedding of $\Delta N(v)$: $\{\textbf{h}_{u}^{l-1}\mid u$$\in$$\Delta N(v)\}$.
		\Ensure  New aggregation embedding $\textbf{a}^l_v$ and vertex embedding $\textbf{h}^{l}_v$.
        \For{each $u$ in $N(v) \cup \Delta N(v)$}  \algcomment{//\textbf{msg\_local}}    
        \State $at_{uv}$=$\texttt{exp}\big(\texttt{LeakyReLU}(a^l([W^l \textbf{h}^{l-1}_v ||W^l {\textbf{h}}^{l-1}_{u})\big)$ 
        \EndFor
        \State $at\_sum_v$= $\textbf{\texttt{sum}}\big( \big\{at_{uv}\mid u$$\in$$N(v) \cup \Delta N(v)\big\}\big)$  \algcomment{//\textbf{nbr\_ctx}}
        \For{each $u$ in $N(v) \cup \Delta N(v)$}  
        \State $a\_score_{uv}$=$\frac{at_{uv}}{at\_sum_v}$ \algcomment{//\textbf{msg\_cbn}}    
        \EndFor
        \State ${\textbf{a}}^{l}_v$$=$$\textbf{\texttt{sum}}\big( \big\{a\_score_{uv} W^l\textbf{h}^{l-1}_u |u\in N(v)\cup\Delta N(v) \big\}\big)$  \algcomment{//\textbf{aggregate}}
        \State  $\textbf{h}^{l}_v$ = \func{elu}(${a}^{l}_v$) \algcomment{//\textbf{update}}
	\end{algorithmic} 
\end{algorithm}

\begin{algorithm}[t]\small
	\caption{Incremental \RTEC for GAT}  
	\label{alg:gat2}  
	\begin{algorithmic}[1]
		\Require  
         Input vertex $v$; Affected neighborhood $\Delta N(v)$; The new embedding of $\Delta N(v)$: $\{\textbf{h}_{u}^{l-1}\mid u$$\in$$\Delta N(v)\}$, The old embedding of $\Delta N(v)$: $\{\textbf{h\_old}_{u}^{l-1}\mid u$$\in$$\Delta N(v)\}$; The old neighbor aggregation embedding $\textbf{a}^{l}_v$; The old neighbor-wise context: attention sum $at\_sum_v$;

		\Ensure  New aggregation embedding $\textbf{a}^l_v$ and vertex embedding $\textbf{h}^{l}_v$.
        
        \State $at\_sum\_old_v$ = $at\_sum_v$
        \For{each $u$ in $\Delta N(v)$}  
        \State $at\_old_{uv} =at_{uv}$
        \State $at_{uv}$=$\texttt{exp}\big(\texttt{LeakyReLU}(a^l([W^l \textbf{h}^{l-1}_v ||W^l {\textbf{h}}^{l-1}_{u})\big)$ \algcomment{//\textbf{msg\_local}}
        \EndFor
        \State $at\_sum_v$= $\textbf{\texttt{sum}}\big( at\_sum_v, \big\{at_{uv}$$-$$at\_old_{uv}\mid u$$\in$$\Delta N(v)\big\}\big)$  \algcomment{//\textbf{nbr\_ctx}}
        \State $\hat{\textbf{a}}^{l}_v$ = $\textbf{a}^{l}_v \cdot at\_sum\_old_v$ \algcomment{//\textbf{msg\_cbn}$^{-1}$} 
        \State $\hat{\textbf{a}}^{l}_v$$=$$\textbf{\texttt{sum}}\big( \hat{\textbf{a}}^{l}_v, \big\{at_{uv} W^l\textbf{h}^{l-1}_u$$-$$ at\_old_{uv} W^l\textbf{h\_old}^{l-1}_u$$|u$$\in$$\Delta N(v) \big\}\big)$  \algcomment{//\textbf{agg}}
        \State $\textbf{a}^{l}_v=\frac{\hat{\textbf{a}}^{l}_v}{at\_sum_v}$ \algcomment{//\textbf{msg\_cbn}}
        \State  $\textbf{h}^{l}_v$ = \func{elu}(${a}^{l}_v$) \algcomment{//\textbf{update}}
	\end{algorithmic} 
\end{algorithm}

\Paragraph{Graph Attention Network (GAT)}
The GAT model, as defined in Equations~\ref{eq:gat:message} and~\ref{eq:gat:agg}, computes the parameterized attention value for each edge and applies the \texttt{softmax()} operation across all neighbors to obtain normalized importance scores to distinguish important neighborhoods. Under our fine-grained abstraction, the \texttt{softmax()}-based message computation (Equation~\ref{eq:gat:message}) can be naturally decomposed into three components, as illustrated in Algorithm~\ref{alg:gat1}.  
First, \textbf{msg\_local()} computes the raw attention logits for each edge (Lines 1–2). Then, \textbf{nbr\_ctx()} performs attention summation by aggregating logits across all neighbors (Line 3). Finally, \textbf{msg\_cbn()} normalizes the attention values by dividing each raw score by the summed value (Line 5).  
The aggregation and update steps remain consistent with Equation~\ref{eq:gat:agg}, where \textbf{f\_nn}($\textbf{h}^l_v$) = $W^l\textbf{h}^l_v$. Therefore, the full-neighbor \RTEC execution using Algorithm~\ref{alg:gat1} yields results equivalent to the original formulation in Equations~\ref{eq:gat:message} and~\ref{eq:gat:agg}.

Algorithm~\ref{alg:gat2} outlines the workflow for incremental computation using the decoupled operators in Algorithm~\ref{alg:gat1}. In addition to the basic inputs, it leverages historical neighbor context and aggregated result, i.e., the attention sum $at\_sum_v$ and the neighborhood aggregation embedding $\textbf{a}^l_v$.  
First, the algorithm computes the raw attention value for each edge in the affected neighborhood (Lines 2–4) while preserving the outdated attention values for later correction (Line 3). Then, \textbf{nbr\_ctx()} updates the old attention sum $at\_sum_v$ to the new one by applying the delta value of affected neighborhoods (Line 5).  
Next, the outdated attention sum $at\_sum\_old_v$ is removed from the existing aggregation embedding $\textbf{a}^l_v$ through the multiplication (Line 6), i,e., the inverse operation of division. This is possible because the attention normalization of \textbf{msg\_cbn()} is distributive over summation, i.e., \(\texttt{sum}(\frac{m_1}{at\_sum_v}, \frac{m_2}{at\_sum_v}) = \frac{\texttt{sum}(m_1, m_2)}{at\_sum_v}\). The algorithm then performs a partial update to the intermediate embedding $\hat{\textbf{a}}^l_v$ using the new local attention values and the input embeddings of affected neighbors (Line 7). Finally, the updated aggregation embedding is recomputed using the new attention sum and passed to the \textbf{update} operation (Lines 8–9). The incremental computation yields outputs equivalent to the original GAT formulation.

\Paragraph{More examples} Many commonly used GNN models that are treated as non-incremental or lack accuracy guarantees in existing frameworks~\cite{ripple_arxiv_2025,INKSTREAM_ARXIV_2023} can benefit from incremental \RTEC~\cite{GNNGOODSURVEY_TPAMI_2024}. Table~\ref{tab:alg:taxomi} presents the decomposed formulations of ten representative models. Among them, CommNet and MoNet are naturally compatible with incremental processing due to their use of edge-wise message functions and \texttt{sum} aggregation, whereas the remaining models can be incrementalized with our design. The GCN model can be incrementalized by decomposing the degree normalization into: \(\frac{1}{\sqrt{d_u}}\), \(d_v\), and \(\sqrt{\cdot}\), which correspond to \textbf{msg\_local()}, \textbf{nbr\_ctx()}, and \textbf{msg\_cbn()}, respectively. GraphSAGE and PinSAGE employ the non-associative \texttt{mean()} aggregation, but \texttt{mean()} can be decomposed into a \texttt{sum()} followed by division by the destination vertex’s degree, enabling full incremental processing. In contrast, models such as GAT, G-GCN, and A-GNN incorporate destination vertex embeddings into edge-wise message computation and therefore employ conditional incremental processing. Beyond homogeneous GNNs, heterogeneous models designed to handle graphs with multiple edge types (e.g., GCN and GAT) can also be incrementalized by processing each edge type independently and merging results in the final step. \revision{It is worth noting that \sysname can also support multi-hop aggregation variants of these models~\cite{FLEXGRAPH_EUROSYS_2021}, since the indirect-aggregation across multi-hop neighborhoods can be viewed as adding temporal edges, which violates the condition.}

\section{\textbf\texttt{NeutronRT} System}
\label{sec:system}

{
Incremental RTEC requires caching intermediate embeddings across layers, making full in-GPU processing impractical for large graphs. We introduce \sysname, a CPU–GPU co-processing system that computes embeddings on GPUs while storing intermediate results and graph data in the larger CPU memory. Figure~\ref{fig:sec3:arch} shows an overview.
}

\subsection{Incremental Computation Engine}
\label{sec:comp_engine}
\Paragraph{Backend and operator implementation} \sysname implements the incremental \textbf{\RTEC} engine on top of DGL~\cite{DGL_ARXIV_2019}, a widely adopted graph learning library. It leverages DGL's underlying subgraph representation to maintain the computation graph and its highly optimized GNN operator implementations for efficient GPU execution. \sysname extends DGL's programming interface to support incremental computation operators. For the \textbf{aggregate} and \textbf{update} operations, it directly reuses DGL's native implementations. For message computation, \sysname introduces four new user-defined functions to support the \textbf{msg\_local}(), \textbf{nbr\_ctx}(), \textbf{msg\_cbn}(), and \textbf{msg\_cbn}$^{-1}$() operations as defined in Algorithm~\ref{alg:inc_comp}. We now demonstrate an example using the GAT model:

{\footnotesize
\begin{lstlisting}[caption={GAT Implementation}, label={lst:gat}]
def msg_local(G, mlc_old: Tensor, h_src: Tensor, h_dst: Tensor):
    m = torch.concat(W(h_src), W(h_dst))
    return mlc_old, torch.exp(LeakyReLU(a(m), 0.2))

def nbr_ctx(nct_old: Tensor, mlc: list[Tensor]):
    return nct_old, torch.sum(mlc)

def msg_cbn(nct: Tensor, mlc: Tensor):
    return nct / mlc

def msg_cbn_r(a_dst: Tensor, nct: Tensor):
    return a_dst * nct
\end{lstlisting}
\vspace{-0.1in}
}

\begin{figure}
	\centering
    \vspace{-0.1in}
	\includegraphics[width=1.0\linewidth]{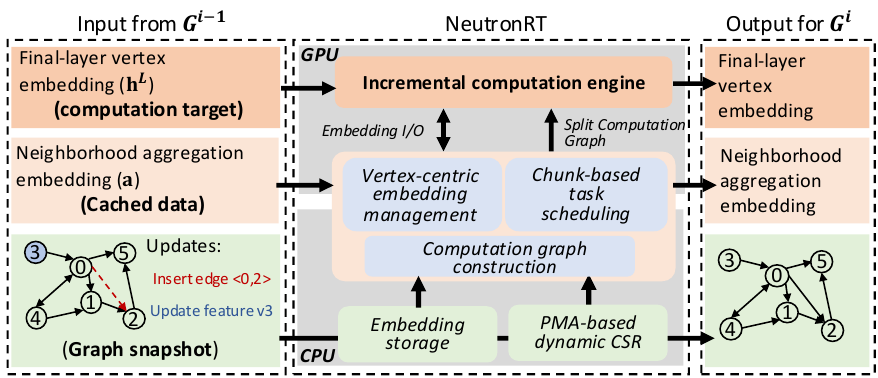}
    \vspace{-0.25in}
	\caption{\sysname system overview.}
    \vspace{-0.13in}
 	\label{fig:sec3:arch}
\end{figure}

{
NeutronRT introduces minimal programming overhead, as most GNN models follow a small set of common patterns (Table~{II}). To apply NeutronRT to a custom GNN model, users first check whether the aggregation function satisfies associativity for incremental accumulation, or can be transformed to do so by removing constant factors. They then examine the edge-level computation. Models with linear and constant edge computations can be handled similarly to the GCN decomposition shown in the third row. More complex models (e.g., GAT, shown in the 7th row) are supported when both aggregation and non-aggregation components can be decomposed into linearly accumulable forms that preserve distributivity with respect to the aggregation. We demonstrate this process using GAT as an example in Listing~{\ref{lst:gat}}.

}

{
\Paragraph{LLM-assisted programming} 
\sysname provides automated operator decomposition and applicability-condition verification tools powered by LLMs and Satisfiability Modulo Theories (SMT) solvers \cite{z3}, thus reducing engineering efforts. Specifically, we construct an external knowledge base that maps GNNs’ original formulations to their implementations in \sysname’s API, along with applicability conditions. This enables users to automatically verify the required conditions and generate incremental programs for any potential GNN model directly from its original formulation via LLMs. To ensure the reliability of LLM-generated programs, \sysname leverages the SMT solvers~\cite{powerlog_sigmod_2020} to detect inconsistencies between incremental and original results of the generated programs, thereby guaranteeing correctness.
}

{
\Paragraph{Computation Graph Construction}
\label{sec:task_scheduling}
The computation graph is constructed via parallel GPU BFS starting from updated vertices and edges, as shown in Algorithm~\ref{alg:affect_region}. At each layer, the affected edge set $E_{{curr}}$ consists of the updated edges and the outgoing edges of affected vertices in the previous layer (Line 3). The destination vertices of $E_{{curr}}$ then form the affected vertex set for the next layer (Line 4). For constrained models (Section~\ref{sec3:condition}), incoming edges of affected destination vertices are additionally included to ensure correctness (Line 5-7). Finally, the resulting per-hop subgraph is appended to the DGL computation graph for execution (Line 8).

\Paragraph{Packed-memory-array-based CSR structure}
To continuously and efficiently support large-scale dynamic graphs, \sysname adopts a CPU-resident packed memory array (PMA)-based dynamic CSR representation~\cite{VCSR_CCGRID_2022} for graph storage. The PMA-CSR stores all vertex neighborhoods in a single packed array, enabling compact storage and efficient neighborhood access to each vertex. Neighborhoods of different vertices are separated by adaptively balanced gaps to accommodate dynamic edge insertions. We refer interested readers to~\cite{VCSR_CCGRID_2022,Grapin_VLDB_2025} for more details.
}

\vspace{-0.04in}
\subsection{Out-of-Memory Embedding Management} 
\label{sec:cpu_gpu_hybrid}
When input features and intermediate embeddings fit in GPU memory, they are fully cached for fast access; otherwise, they remain in CPU memory and are accessed on demand.

\Paragraph{Embedding-centric data migration}
Incremental computation involves sparse vertex embedding accesses, resulting in complex and inefficient CPU–GPU communication. To mitigate this, \sysname adopts the GPU-directed zero-copy memory access to directly read sparsely distributed vertex embeddings from CPU memory with optimized PCIe bandwidth utilization~\cite{PYTORCHDIRECT_VLDB_2021,EMOGI_VLDB_20}. After computation, \sysname group all update embeddings and write them back in parallel. \sysname omits GPU-side intermediate caching, as it offers limited benefit under tight memory constraints.

\Paragraph{Recomputation-based embedding storage optimization} Maintaining both the \textbf{aggregate} output neighbor embedding \(\textbf{a}^l\) and the \textbf{update} output layer embedding \(\textbf{h}^l\) results in doubled CPU memory overhead. We observe that the \textbf{update} computation, transforming \(\textbf{a}^l\) into \(\textbf{h}^l\), typically incurs very little overhead, as it involves only vertex-wise neural network computation. Therefore, to reduce the CPU memory consumption of embedding storage for large graphs, we choose to cache and access only the neighbor embedding \(\textbf{a}^l\), and recompute \(\textbf{h}^l\) on the GPU at runtime.

{
\Paragraph{Out-of-CPU embedding management}
\sysname core design assumes that the intermediate embeddings can be fully cached in CPU memory for efficient incremental processing. When the graph exceeds the memory capacity of a single machine, \sysname can fall back to selectively caching embeddings of high-degree vertices and recomputing the rest on demand. Although this heuristic retains embeddings with higher reuse potential~\cite{DegreeSort_VLDB_2024}, even removing a small fraction (e.g., 10\%) of cached embeddings can lead to over an order-of-magnitude slowdown. This behavior stems from the inherent space-time trade-off of incremental processing. 
}

\begin{algorithm}[t]
\caption{Computation graph construction}
\label{alg:affect_region}
\begin{algorithmic}[1]\small
\Require Computation graph $\{CG_1 ...CG_{L}\}$

\Ensure Computation graph $\{CG_1 ...CG_{L}\}$

\State $V_{curr} \leftarrow V_{upd}$
\For{hop $l \in \{1, \cdots, L\}$} \textbf{in parallel}
    \State $E_{curr}\leftarrow E_{upd} \cup  \{ <u,v> \mid u \in V_{curr}\} $
    \State $\{V_{dst}, V_{src}\}\leftarrow E_{curr}$; $E_{recomp}\leftarrow \emptyset$; $V_{curr}\leftarrow \emptyset$
    \For{each $v \in V_{dst}$} \textbf{in parallel}
    \If{ $\text{constraint\_model}$ \texttt{and} $v \in V_{src}$}
    \State $E_{recomp}\leftarrow E_{recomp}\cup\texttt{inEdges}(v)$
    \EndIf
    \EndFor
    \State $CG \leftarrow$\texttt{construct\_graph}($V_{dst}$,$\{E_{curr}\cup E_{recomp}\}$, l)
\EndFor
\end{algorithmic}
\end{algorithm}

\subsection{Chunked Task Scheduling with Shard Embedding Reuse}
\label{sec4:chunk_scheduling}
 Considering the memory requirements of large computation graphs may exceed the capacity of a single GPU, \sysname adopts a chunked task scheduling approach that partitions each layer's computation graph into smaller chunks that fit within GPU memory. During partitioned processing, a vertex may appear in the neighborhoods of multiple chunks, causing its embedding to be transferred multiple times within the same layer. To reduce such redundant transfers, \sysname incorporates an inter-chunk embedding reuse mechanism~\cite{HONGTU_SIGMOD_2024}.
Specifically, \sysname precomputes neighborhood intersections across adjacency chunks and caches the shared embeddings in an intermediate GPU buffer, enabling reuse across chunks within the same layer. In practice, the chunk size is chosen to be as large as possible within the GPU memory budget, as larger chunks reduce data movement and increase cross-chunk overlap, facilitating embedding reuse.

{

\subsection{On-Demand Embedding Computation}
In some online applications, only a small set of query vertices is requested at a time, requiring their embeddings to be computed on demand. This mode, termed On-Demand Embedding Computation (\texttt{ODEC}), can be viewed as a special case of \RTEC that processes only the K-hop subgraph induced by the queried vertices for serving online queries in real time~\cite{netflix2022personalization,meituan}. From an execution perspective, \texttt{ODEC} similarly incurs redundant computation over unaffected edges within the K-hop subgraph. As the computation logic is unchanged, it can benefit from incremental execution and is efficiently supported by \sysname. Specifically, \sysname constructs the \texttt{ODEC} computation graph by intersecting the affected subgraph with the query-induced K-hop subgraph and executes it using the same task scheduling engine.

}

\section{Experimental Evaluation}
\label{sec:expr}

\label{sec:expr:setup}

\Paragraph{Environments} 
The experiments are conducted on a GPU server with 2 Intel(R) Xeon(R) Silver 4316 CPUs, 512GB DRAM, and one NVIDIA A6000 ($48$GB) GPU connected to the CPU via the PCIe 4.0@32GB/s. The server runs Ubuntu $20.04$ OS with GCC-$9.4.0$, CUDA $11.3$ and PyTorch $1.13.0$.

\begin{table}[!t]
\vspace{-0.1in}
	\caption{Dataset description. $\#\mathbb{F}$, $\#\mathbb{H}$, and $\#\mathbb{L}$ indicate the number of features, hidden dimension, and labels, respectively.}
	\vspace{-0.1in}
	\label{tab:Dataset}
	\centering
	\footnotesize
	{\renewcommand{\arraystretch}{1.1}
    \setlength{\tabcolsep}{2pt}
	\begin{tabular}{l r r c c c c c}
		\hline
		
		\hline
		scale&{\textbf{Dataset}} &
		{\textbf{|V|}} &
		{\textbf{|E|}}  &
		{\textbf{\#$\mathbb{F}$}}&
		{\textbf{\#$\mathbb{H}$}}&
		{\textbf{\#$\mathbb{L}$}}&
		{\textbf{Type}}\\
		\hline
        
		&{ogbn-arxiv \cite{OGB_DATASET} (AX)} & 0.17M & 1.2M &128&256&40& citation\\
		small&{reddit \cite{REDDIT_2017} (RD)} & 0.23M & 114M &602&256&41& post-to-post\\		
		&{ogbn-products \cite{OGB_DATASET} (PT)}  &2.4M& 62M&100 &256&47&co-purchasing\\       
        \hline
        &{Twitter \cite{TWITTER_WWW_2010} (TW)}&	41M &1.5B&128&128&64& social media\\
        
		large&{ogbn-paper \cite{OGB_DATASET} (PR)}& 111M &1.6B&128&128&172&citation network\\

        &{friendster} \cite{FS}(FS)& 65.6M & 2.5B& 128& 128& 64&social media\\

		\hline
		
	\end{tabular}
	}

\vspace{-0.1in}
\end{table}

\begin{figure*}[!t]
\vspace{-0.05in}
	\centering
	\includegraphics[width=\linewidth]{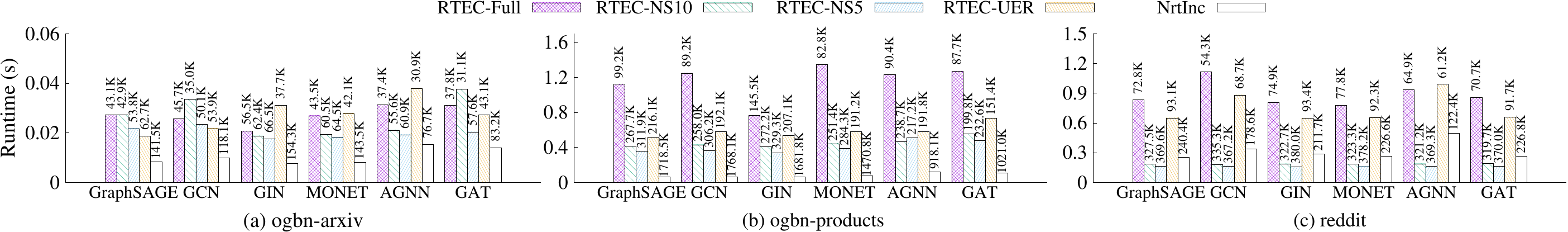}
 	\vspace{-0.2in}
	\caption{Average Response time and throughput (on the top of each bar) comparison across fix models with GPU in-memory processing.}
	\label{fig:sec5:im-overall}
 \vspace{-0.1in}
\end{figure*}

\begin{figure*}[!t] 
\flushleft
\hspace{-0.05in}
\begin{minipage}{.31\linewidth}
    \begin{center}
    \includegraphics[width=\linewidth]{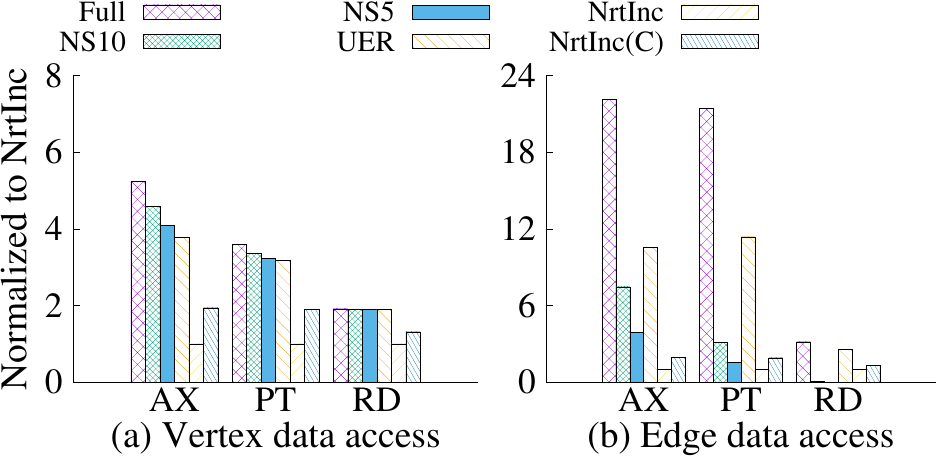}
    \vspace{-0.2in}
    \caption{Data access volume on small graphs.} 
    \label{fig:im_access}
    \end{center}
\end{minipage}
\begin{minipage}{.25\linewidth}
    \begin{center}
    \includegraphics[width=\linewidth]{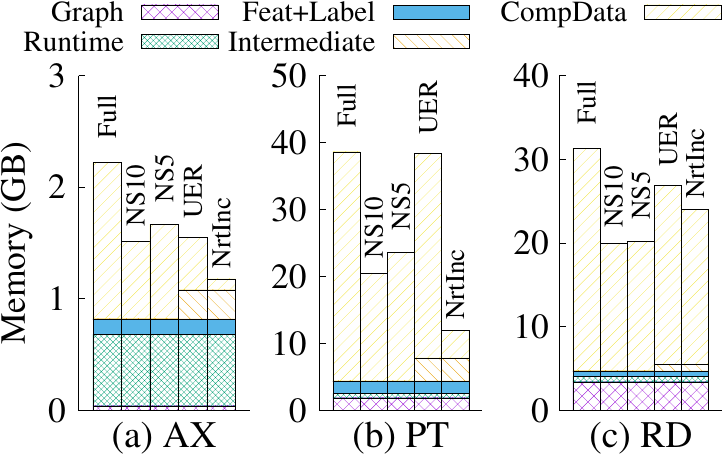}
    \vspace{-0.2in}
    \caption{Memory consumption.} 
    \label{fig:im_m_bd}
    \end{center}
\end{minipage}
\begin{minipage}{.42\linewidth}
    \begin{center}
	\includegraphics[width=\linewidth]{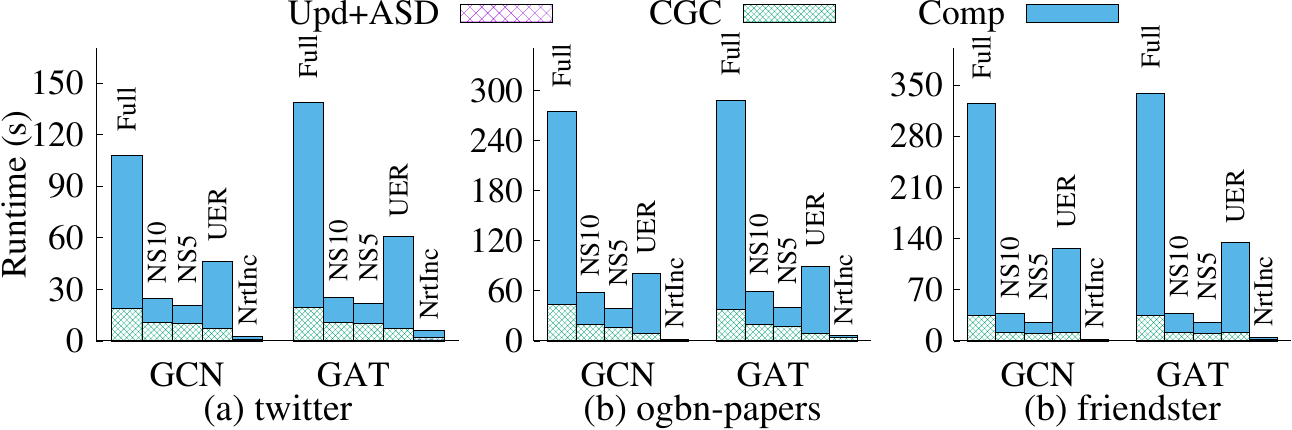}
 	\vspace{-0.2in}
	\caption{Runtime and the breakdown for large graphs.}
	\label{fig:om_p_bd}
    \end{center}
\end{minipage}
 \vspace{-0.2in}
\end{figure*}

\Paragraph{Datasets and workloads}
Table \ref{tab:Dataset} shows the major parameters of the used real-world graphs. Four of them use their attached properties and default training/validation/ test splits. We generate random features, labels, and edge time stamps following \cite{DYNAGRAPH_GRADES_2022} for the Twitter and Friendster graph and randomly allocate 25\%, 25\%, and 50\% of the vertices for training, validation, and testing, respectively. In the experiments, we use edge insertion/deletion hybrid workload~\cite{RisGraph_SIGMOD_2021,REDDIT_2017,OGB_DATASET}. The number of graph updates is controlled using the number of edges, as the power-law property leads to significant variations in edge counts across different vertices. By default, the update batch size is set to 0.01\% of \(|E|\) for small graphs and 0.001\% of \(|E|\) for large graphs.

\Paragraph{GNN models}
We evaluate six models with diverse computation patterns. GCN, GraphSAGE, MoNet, and GIN can be fully incrementalized. AGNN and GAT are constrained incremental models.

{
\Paragraph{Baselines}
We compare \sysname employing incremental \RTEC processing (denoted by \NrtInc) against \RTECFull, \RTECNS \cite{Helios_PPOPP_2025}, and \RTECUER \cite{Grapher_WWW_2024}. $\lambda$Grapher~\cite{Grapher_WWW_2024} optimizes \RTEC on serverless platforms by reusing intermediate results with \RTECUER, while Helios~\cite{Helios_PPOPP_2025} accelerates dynamic graph sampling for \RTECNS on memory–computation decoupled architectures. As both use different deployment settings from ours, we reimplemented their core approaches in \sysname using DGL’s built-in sampling engine and \sysname’s task scheduling mechanism to ensure a fair comparison.  For \RTECNS
We adopt sampling sizes of 5 and 10~\cite{DGL_ARXIV_2019, Helios_PPOPP_2025}.  
The comparison against existing Incremental framework is given in Section \ref{sec:expr:sensitive}. The chunk size for memory-efficient task scheduling is set to 8192 to ensure each chunk fits into GPU memory. In our evaluation, we report only the result of embedding computation and exclude training time, which is performed offline. All reported results are averaged over 20 batches for consistency. We exclude recent model- or channel-pruning-based approximation methods~\cite{REALTIMEGNN_VLDB_2021} from the evaluation, as they are orthogonal to our structure-centric approach. Applying such techniques to both NeutronRT and \RTEC-Full would yield similar computation reductions and accuracy trade-offs.

}

{
\subsection{Accuracy Analysis of Incremental \RTEC}
\label{sec:expr:accuracy}

In this section, we show that incremental \RTEC (\textbf{NrtInc}) matches the accuracy of full \RTEC, outperforms \textbf{RTEC-NS}~\cite{HELIX_VLDB_2018} and \textbf{MTEC-Period}~\cite{meituan,inferturbo_ICDE_2023}, and incurs only minor accuracy loss compared with the theoretically optimal but impractical \textbf{MTEC-Optimal} (real-time model retraining and embedding recomputation).
The experiments are conducted on three real-world graphs using node classification tasks\footnote{Reddit and ogbn-arxiv include native timestamps, while ogbn-products adopts synthesized timestamps following~\cite{DYNAGRAPH_GRADES_2022}. All experiments employ the GraphSAGE model and DGL's default training configuration (200 epochs, 2 layers, a batch size of 2048, and a sampling fanout of 10). 
For the \textbf{MTEC-Period} method, both training and inference are conducted on 90\% of the old graph. The \textbf{MTEC-Optimal} trains and infers on the updated graph. In contrast, \RTEC trains on 90\% of the old graph and performs inference on the updated graph. 
}.
}

{

\color{blue}
\begin{table}
\footnotesize
\vspace{-0.05in}
\renewcommand{\arraystretch}{1.1}
\setlength{\tabcolsep}{0.6pt}
	\caption{Accuracy comparison of different GNN Inference approaches with a 2-layer GraphSage on three real graphs.}
\vspace{-0.05in}
	\label{tab:sec2:motivation}
\begin{tabular}{l |r| r | r| r l  r  r r } 
\hline

\multirow{2}*{Approach}& \multicolumn{4}{|c}{ \textbf{Accuracy comparison of 5 runs}}\\
\cline{2-5}
 &ogbn-arxiv & ogbn-product & reddit & ogbn-paper \\ \hline
\textbf{MTEC-Optimal}&\best{68.57\%$\pm$0.05}\%	& \secondbest{72.84\%$\pm$0.06\%}	&\best{96.68\%$\pm$0.03\%} &\best{65.03\%$\pm$0.06\%}\\

\textbf{MTEC-Period} &67.10\%$\pm$0.15\%&72.28\%$\pm$0.07\%	& 94.73\%$\pm$0.04\% & {63.37\%$\pm$0.01\%} \\
\hline

\multirow{1}*{\textbf{RTEC-NS5}}&66.70\%$\pm$0.11\%	&71.23\%$\pm$0.07\% &	95.47\%$\pm$0.16\% &{63.12\%$\pm$0.10\%}  	\\

 \multirow{1}*{\textbf{RTEC-NS10}}   &67.74\%$\pm$0.09\%
 &72.30\%$\pm$0.09\% &	95.92\%$\pm$0.20\% 	&{63.67\%$\pm$0.14\%} \\
 
\multirow{1}*{\textbf{RTEC-NS20}}    &{67.91\%$\pm$0.11\%} & {72.72\%$\pm$0.08\%} &{96.21\%$\pm$0.11\%} &{64.43\%$\pm$0.08\%} 	\\
\hline

\textbf{RTEC(\NrtInc)}$^1$ & \secondbest{68.20\%$\pm$0.11\%} &	\best{72.86\%$\pm$0.08\%}	&\secondbest{96.58\%$\pm$0.03\%} &\secondbest{64.88\%$\pm$0.04\%} \\
\hline
\end{tabular}

\vspace{0.03in}
{
[1] \RTEC and \NrtInc achieve nearly identical results across GNN models. The mean squared error (MSE) between the final-layer embeddings produced by RTEC-Inc and RTEC-Full is below $10^{-4}$. We therefore report only the GraphSAGE results as a representative case, using a single row.

}
\vspace{-0.1in} 
\end{table}
}

\Paragraph{Accuracy Results} 
As shown in Table~\ref{tab:sec2:motivation}, \textbf{MTEC-Period} can lag behind the \textbf{MTEC-Optimal} solution by up to 2\%, which may affect hundreds of thousands of users in large-scale applications. In contrast, \textbf{NrtInc} improves accuracy over \textbf{MTEC-Period} by 1.18\% on average, while keeping the gap to the optimal solution within 0.15\%. These results indicate that embedding quality is more sensitive to graph structure changes than to model freshness.
Sampling-based methods perform even worse than \textbf{MTEC-Period} with a fanout of 5, as the discarded subgraph is much larger than the actual update region. Although increasing the fanout improves accuracy, performance remains inferior to incremental \RTEC even at a fanout of 20, while also incurring significantly lower computation efficiency (Section~\ref{sec:expr:im}).

{

\vspace{-0.05in}
\subsection{Performance of In-Memory Processing}
\label{sec:expr:im}
In this section, we compare \sysname (\NrtInc) against all baselines with three small graphs and six GNN models. The intermediate embedding is maintained in the GPU.

{
\Paragraph{Runtime and throughput results}
Figure~\ref{fig:sec5:im-overall} reports the response time per update batch under different configurations. On arXiv and Products graphs, {\NrtInc} achieves speedups ranging from 2.6x to 3.3x (10.0x to 19.8x), 1.9x to 3.1x (3.8x to 6.9x), 1.7x to 3.4x (4.2x to 5.8x), and 2.2x to 4.1x (4.8x to 9.2x) over \RTECFull, \RTECNSTen, \RTECNSFive, and \RTECUER, respectively. On the Reddit graph, \NrtInc outperforms \RTECFull by 1.9x-3.3x  \RTECUER by 2.0x-2.6x, but is 0.3x to 0.7x slower than \RTECNS. The throughput, shown above each bar, is computed as $\frac{|updated\_edges|}{response\_time}$, which is inversely proportional to response time. Under the default setting (0.01\% edges per batch), \sysname achieves 76.7K–154.3K, 918.1K–1768.1K, and 178.6K–378.2K edge updates per second on the three graphs, respectively.
}

\Paragraph{Access volume reduction}
Figure~\ref{fig:im_access} compares vertex and edge access volumes across different approaches. All models exhibit similar access patterns, except constrained models (e.g., GAT and AGNN), which incur slightly higher accesses, denoted by \NrtIncC. Overall, \NrtInc consistently reduces access volumes compared to \RTECFull. Compared to \RTECNS, \NrtInc incurs higher edge accesses but significantly reduces vertex accesses. On Reddit, the high average degree and small vertex set cause sampled subgraphs to cover most vertices, diminishing the effectiveness of sampling. Importantly, reductions in edge and vertex accesses are not always aligned, and vertex access volume has a non-negligible impact on compute efficiency. As a result, although \RTECNS reduces edge accesses by up to 97\% relative to \NrtInc, its speedup is limited to at most 3$\times$. In contrast, \NrtInc reduces both vertex and edge accesses, leading to consistent and substantial performance improvements without compromising accuracy.

\Paragraph{Performance of constrained incremental models}
For constrained incremental models (AGNN and GAT), the overhead is generally less than twice that of pure incremental processing, as the recomputed subgraph shares the same structure with the incremental computation graph and contains duplicated vertices and edges. As shown by the \NrtIncC bar in Figure~\ref{fig:im_access}, the recomputation introduces 31\%–94\% more vertex accesses and 33\%–94\% more edge accesses, resulting in a 1.2x–1.7x slowdown compared to \NrtInc. Nevertheless, the average speedup of \NrtIncC over \RTECFull remains substantial, ranging from 2.1x to 10.9x across three graphs.

{
\Paragraph{Memory consumption}
\sysname caches all intermediate embeddings during computation, but this will not increase the memory consumption. Figure~\ref{fig:im_m_bd} shows the total memory usage under in-memory processing (collected via Nsight) along with the breakdown across different components. The \textbf{CompData} component includes the affected subgraph, its associated features, embeddings, intermediate tensors, and GPU memory retained by the Python runtime. We observe that the intermediate result cache of \NrtInc (\textbf{Intermediate}) accounts for 3\% (RD) to 28\% (PT) of the total runtime memory. This is much smaller than the memory saved by reducing the computation graph. Across the three graphs, total memory usage is reduced by 23\% to 69\% (avg. 42\%).
}

\subsection{Performance of Out-of-Memory Processing}

In this section, we evaluate GCN and GAT model on three large graphs, as other models exhibit similar behaviors.

\Paragraph{Runtime and throughput comparison}
The runtime results are presented in Figure~\ref{fig:om_p_bd}. We observe that the processing time of \RTECFull and \RTECUER ranges from 107.8s to 324.3s and from 46.4s to 127.2s, respectively. Although \RTECNS effectively reduces computation, its runtime still exceeds 20 seconds. In contrast, \sysname reduces the per-batch processing time to just a few seconds (2.2s–2.9s). This enables \sysname to effectively handle scenarios with frequent updates on large-scale graphs, while incurring no correctness issues like those found in sampling-based methods (Section~\ref{fig:sec2:comp_amplification2}). Specifically, for the GCN (and GAT) models, \NrtInc achieves speedups ranging from 37.8x to 145.8x (22.7x to 76.8x), 8.6x to 26.1x (4.2x to 8.6x), 7.3x to 17.1x (3.6x to 6.6x), and 16.3x to 57.1x (10.0x to 30.4x) over \RTECFull, \RTECNSTen, \RTECNSFive, and \RTECUER, respectively. GAT incurs higher computation overhead than GCN, primarily due to increased computation and data transfer for the recomputation part. In out-of-memory setting, \sysname achieves throughput ranging from 681.8K–872.5K edge upd/s among the three large graphs.

\Paragraph{Performance Breakdown}
Figure~\ref{fig:om_p_bd} presents the performance breakdown for the three large graphs. In the out-of-memory configuration, computation becomes the dominant bottleneck, accounting for 82\% to 89\% of the total runtime in the \RTECFull setup. This overhead stems not only from the increased volume of graph accesses but also from additional scheduling and engine launch costs introduced by chunked graph processing. Computation graph construction contributes a further 11\% to 18\% of the total runtime, while the overhead of graph updates and affected region detection is negligible. The \RTECNS method effectively reduces computation through neighborhood dropping. However, its benefit in reducing communication is limited, and its overall computation cost remains higher than that of the incremental approach. In contrast, \NrtInc significantly reduces both computation graph construction and overall computation costs, by 90\%–98\% and 95\%–99\%, respectively.  With the \NrtInc configuration, Upd+ASD, CGC, and Comp account for 4\%–18\% (average 10\%), 16\%–58\% (average 36\%), and 37\%–68\% (average 53\%) of the total execution time, respectively.

\begin{figure}[!t] 
\flushleft
\hspace{-0.05in}
\begin{minipage}{.58\linewidth}
    \begin{center}
    \includegraphics[width=\linewidth]{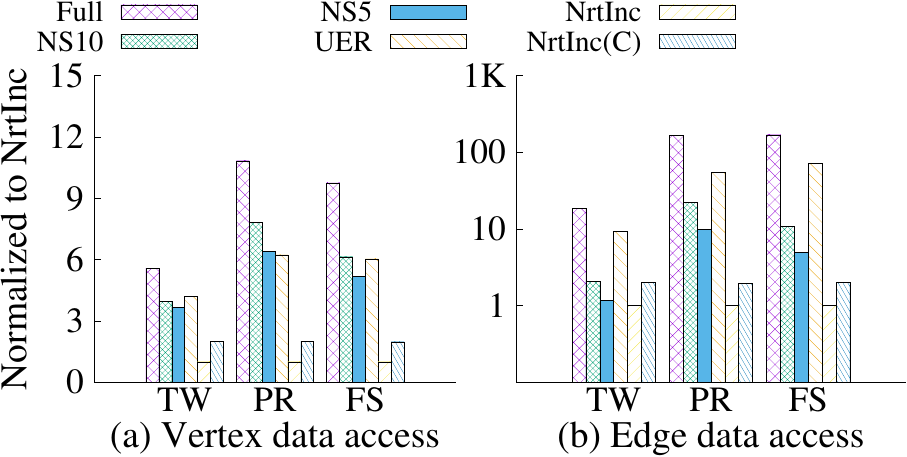}
    \vspace{-0.2in}
    \caption{Access volume on large graphs.} 
    \label{fig:om_access}
    \end{center}
\end{minipage}
\begin{minipage}{.4\linewidth}
\captionof{table}{Edge-access reduction over high-, mid-, and low-degree vertices.}
	\label{tab:expr:access_breakdown}
	\centering
	\footnotesize
	\renewcommand{\arraystretch}{1.0}
\setlength{\tabcolsep}{1.5pt}
	\begin{tabular}{r |r |r |r}			
		\hline
		&{Top20\%}&
		{Mid30\%}&
		{Bot50\%}\\	
        \hline
        TW&85.0\%& 9.3\%& 5.7\%\\
        FS&79.2\%& 17.2\%& 3.6\%\\
        PR&71.9\%& 18.5\%&9.6\%\\
		\hline
	\end{tabular}
\end{minipage}
\end{figure}

{
\Paragraph{Access volume reduction}
As shown in Figure~\ref{fig:om_access}, we extend the redundancy analysis to large-scale graphs and observe that \NrtInc consistently achieves substantial reductions in both vertex and edge access volumes compared to existing methods. On billion-scale datasets, \NrtInc reduces vertex accesses by $3.6\times$–$10.8\times$ and edge accesses by $4.9\times$–$168.5\times$. On the Twitter (TW) dataset, sampling-based methods achieve edge access volumes comparable to \NrtInc. This behavior is mainly attributed to the stronger power-law property of TW, where a small fraction of hub vertices with extremely high degrees dominate neighborhood expansion, causing the incremental computation subgraph to grow rapidly even for small updates. Table~\ref{tab:expr:access_breakdown} further shows a finer-grained access volume breakdown by vertex degree percentiles. Across all datasets, high-degree vertices account for 72\%–85\% of the total access reduction, while the bottom 50\% contribute less than 10\% despite comprising half of the vertices. This result indicates that the effectiveness of \NrtInc primarily arises from eliminating redundant edge accesses around high-degree vertices.
}

{
\begin{table}[t]
	\caption{CPU memory consumption for large graphs (GB).}
	\vspace{-0.1in}
	\label{tab:expr:mem}
	\footnotesize
	\renewcommand{\arraystretch}{1.0}
\setlength{\tabcolsep}{3pt}
	\begin{tabular}{r r r| r r r| r r r}		
		\hline
		
		\multicolumn{3}{c|}{twitter}&
        \multicolumn{3}{c|}{friendser}&
        \multicolumn{3}{c}{ogbn-paper}\\
		\hline
		{Full}&
		{Inc-Naive}&
		{Inc}&
		{Full}&
		{Inc-Naive}&
		{Inc}&
		{Full}&
		{Inc-Naive}&
		{Inc}
        \\
		\hline

        20.99& 83.97& 62.98&
        56.83& 227.32& 170.4&
        33.58& 134.35&100.76
        \\
		\hline
	\end{tabular}
	\vspace{-0.2in}
\end{table}

\Paragraph{CPU memory consumption for large graphs}
As shown in Table~\ref{tab:expr:mem}. \RTECFull and \RTECNS only store the original features, resulting in the lowest memory usage. However, for large-scale graphs, even this can exceed the capacity of a single GPU. The naive \NrtInc stores both the intermediate neighborhood aggregation and vertex embeddings alongside the original feature, leading to a 2.3x–8.2x increase in memory usage. In contrast, the recomputation-based optimization reduces the CPU memory consumption by 19\%–30\% with negligible recomputation cost (smaller than 1\%).
}

\begin{figure*}[t]
	\centering
	\includegraphics[width=\linewidth]{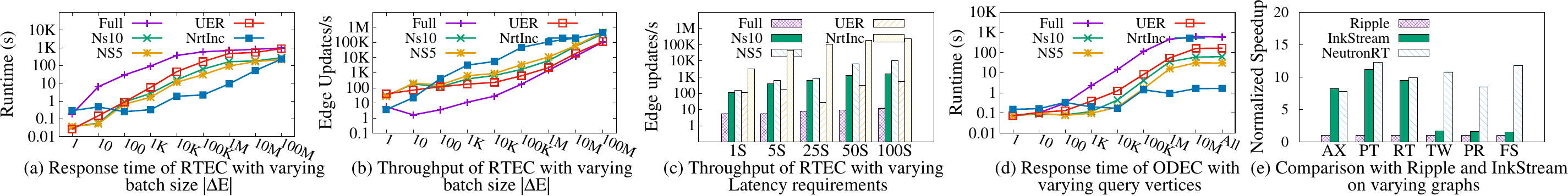}
\vspace{-0.2in}
	\caption{Performance incremental RTEC and ODEC with varying factors.}
	\label{fig:sec5:varying}
    \vspace{-0.15in}
\end{figure*}

\subsection{Sensitivity Analysis}
\label{sec:expr:sensitive}
{
\Paragraph{Performance of RTEC with varying batch sizes ($|\Delta E|$)}
Figure~\ref{fig:sec5:varying}.a–b show the response time and throughput on the ogbn-Paper graph as the number of edge updates $|\Delta E|$ increases from 1 to 100M. When $|\Delta E| < 10$, \NrtInc exhibits slightly lower performance than other solutions due to the overhead of constructing the \textit{affected subgraph}. As $|\Delta E|$ increases, the runtime and throughput of \NrtInc gradually approach those of \RTECFull. Notably, \texttt{NrtInc}’s runtime increases slowly, while its throughput grows rapidly. The performance advantage of \NrtInc becomes significant with moderate update sizes, peaking at a 285.0x speedup over \RTECFull when $|\Delta E| = 1\text{K}$. Beyond this point, the advantage gradually declines, dropping to 3.8x at $|\Delta E| = 100\text{M}$, where the updates account for approximately 16\% of the original edge set. At this scale, \RTECFull almost needs to recompute the entire graph, and \RTECNSFive achieves similar performance to \NrtInc. However, \RTECNSFive suffers from information loss and accuracy degradation, as many affected graph components are omitted during recomputation. \NrtInc maintains stable improvement across varying batch sizes. 
}

{
\Paragraph{Throughput with various latency requirements}
Real-world applications require high throughput (updates/sec) while satisfying latency constraints, such that the updated embeddings are visible to downstream applications within a bounded time. However, practical GNN embedding computation latency may vary significantly under the same batch size due to irregular update propagation and varying neighborhood sizes. We therefore estimate a stable and achievable throughput empirically. For each latency requirement (1\,s, 5\,s, 25\,s, 50\,s, and 100\,s), we gradually increase the batch size from 1 to 10M and evaluate each batch configuration using 20 randomly sampled edge update sets. We select the largest batch size for which all runs satisfy the latency bound and report the corresponding minimal achievable throughput. As shown in Figure~\ref{fig:sec5:varying}.c, \NrtInc consistently achieves orders-of-magnitude higher throughput than competing methods, ranging from 3K edge/s under a 1\,s latency bound to 229K edges/s under a 100\,s latency bound.
}

\Paragraph{Performance of \ODEC with varying query vertices}
{
In real-world \ODEC applications, the number of query vertices ($|V_Q|$) varies dynamically. To evaluate incremental \ODEC under different query sizes, we vary $|V_Q|$ from 1 to the full set of affected vertices (Figure~\ref{fig:sec5:varying}.d), considering an extreme case where all queries are drawn from the affected set. Since \NrtInc incurs no overhead on unaffected vertices, we exclude other baselines in this setting. When $|V_Q|<100$, all methods show similar performance due to the small computation graphs. As $|V_Q|$ increases, the advantage of \NrtInc becomes more pronounced. When all affected vertices are queried (\textbf{$ALL$}), \ODEC reduces to \RTEC and achieves the same performance gains. Although the runtime of \NrtInc does not increase monotonically with $|V_Q|$ due to varying overlap with affected subgraphs, it consistently delivers competitive performance when $|V_Q|\geq 1$K.

}

\begin{table}[h]
	\caption{Performance of {\NrtInc} with various layers.}
	\vspace{-0.1in}
	\label{tab:expr:layers}
	\centering
    \setlength{\tabcolsep}{5pt}
	\footnotesize
	\renewcommand{\arraystretch}{1.0}
	\begin{tabular}{l |r r| r r |r r |r r r}
		
		\hline
		\multirow{2}*{\textbf{Data}}&
		\multicolumn{8}{c}{Normalized Speedup}\\
		\cline{2-9}
		&\multicolumn{2}{c|}{{\texttt{Full}}} &
		\multicolumn{2}{c|}{\texttt{NS5}}  &
		\multicolumn{2}{c|}{\texttt{NS10}}&
		\multicolumn{2}{c}{\texttt{UER}}\\
        \hline
        \textbf{Layer}&2&3&2&3&2&3&2&3\\   
		\hline
		{\textbf{Reddit}}&
        3.3x& 1.8&  
        0.5x& 0.3x& 
        0.5x& 0.3x&
        2.6x& 1.5x\\
\hline
        
		{\textbf{Product}}& 
        8.4x& 3.5x&
        4.5x& 2.2x& 
        4.2x& 2.0x&
        4.6x& 2.4x\\
\hline
{\textbf{Paper}}& 
122.6x	& 20.4x	&
26.2x	& 5.3x	&
17.1x	& 3.4x  &	
35.9x   & 6.6x  \\
		\hline
	\end{tabular}
	\vspace{-0.1in}
\end{table}

{
\Paragraph{Performance with varying layers}
Table~\ref{tab:expr:layers} reports the performance improvement of \NrtInc over other baselines. As the number of GNN layers increases from 2 to 3, the speedup generally decreases, since the size of the \textit{affected subgraph} expands rapidly with depth. This trend is consistently observed across graphs with diverse structural properties, including ogbn-products (power-law), Reddit (high average degree), and ogbn-papers (large-scale, power-law). On dense graphs such as Reddit, sampling-based methods become more competitive at larger depths by uniformly reducing computation at each layer. In practice, GNN models are typically configured with 2 or 3 layers, as deeper models not only incur higher training costs~\cite{BG3_SIGMOD_2024,Helios_PPOPP_2025}, but also suffer from accuracy degradation due to over-smoothing~\cite{oversmooth_arxiv_2023}. Despite the reduced speedup at greater depths, the performance gains accumulated in earlier layers remain substantial, allowing \NrtInc to consistently outperform RTEC-Full and RTEC-UER.

}

{
\Paragraph{Comparison with InkStream and Ripple}
We compare NeutronRT with InkStream~\cite{INKSTREAM_ARXIV_2023}, a CPU–GPU hybrid incremental RTEC system, and Ripple~\cite{ripple_arxiv_2025}, a CPU-based RTEC system. Due to the lack of a public implementation, we implement Ripple's execution logic using DGL’s CPU backend in NeutronRT. To ensure a fair comparison, we focus on the GraphSAGE model, which is supported by all three systems. More complex models, such as GCN and GAT are excluded because they are not supported by the computation model of either Ripple or InkStream. Accuracy evaluation are omitted due to space constraints, as all three systems produce similar results. As shown in Figure~\ref{fig:sec5:varying}.e, NeutronRT and InkStream achieve similar performance on small graphs where all data fits in memory, with speedups ranging from 7.8$\times$ to 12.3$\times$ over Ripple. On large graphs, NeutronRT outperforms InkStream by 5.3$\times$–7.7$\times$, because InkStream executes graph propagation on the CPU, whereas NeutronRT performs the entire computation on the GPU.

}

\section{Related Work}
\label{sec:related}
\Paragraph{GNNs for dynamic Graphs}
Temporal Graph Learning focuses on capturing temporal dependencies and structural evolution in dynamic graphs, typically modeled as either discrete-time snapshots or continuous event streams~\cite{DyGNN_SIGIR_2020, DyRep_ICLR_2019, TE_DyGE_DASFAA_2023, TGL_VLDB_2022,Zebra_VLDB_2023}.  
They generally employ encoder-decoder framework \cite{DCRNN_ICLR_2018} and recursive neural networks with various temporal-spatial encoding approaches~\cite{ROLAND_KDD_2022,TGL_VLDB_2022}. 

\Paragraph{Incremental graph processing}
Recent work on incremental graph processing explores fixed-point semantics and monotonicity in iterative graph analytics to avoid computation on converged vertices~\cite{vora2017kickstarter, mariappan2019graphbolt, gong2021ingress, iTurboGraph}.
However, these techniques do not naturally extend to GNNs, as the redundancy in GNNs primarily stems from accessing unaffected edges, rather than recomputing already converged vertices. Furthermore, GNN computations generally lack monotonicity due to the non-linear nature of neural networks.

\section{Conclusion}

We present \sysname, an efficient and general framework for incrementalizing GNN \RTEC on streaming graphs. By fully reusing historical computation results, \sysname transforms expensive full-neighbor propagation into an equivalent yet far more efficient incremental form through fine-grained operator decoupling and reorganization. This design enables broad generalization across diverse GNN models while rigorously preserving correctness. A lightweight CPU–GPU co-processing framwwork further enhances computation efficiency on billion-scale graphs. Extensive evaluations demonstrate that \sysname delivers substantial speedups over existing \RTEC solutions without compromising model accuracy, making real-time GNN inference both practical and reliable.

\section*{Acknowledgements}
We thank the anonymous reviewers for their constructive comments and suggestions. This work is supported in part by the Ministry of Education AcRF Tier 2 grant, Singapore (T2EP20224-0038), National Natural Science Foundation of China (U2241212, 62461146205), Distinguished Youth Foundation of Liaoning Province (2024021148-JH3/501)

\bibliographystyle{abbrv}
\bibliography{main.bib}  

\end{document}